%% file: main.tex
\documentclass[acmsmall,nonacm]{acmart}
\graphicspath{{Figures/}}

\usepackage{blkarray}
\usepackage{amsthm}
\usepackage{algorithm}
\usepackage{algpseudocode}

\theoremstyle{plain}
\newtheorem{theorem}{Theorem}[section]
\theoremstyle{remark}
\newtheorem{remark}{Remark}[section]

\AtBeginDocument{%
}

\begin{document}

\title{On the Accuracy Limits of Sequential Recommender Systems: An Entropy-Based Approach}

\author{En Xu}
\affiliation{%
  \department{Department of Electronic Engineering}
  \institution{Tsinghua University}
  \city{Beijing}
  \country{China}
}
\email{xuen@mail.tsinghua.edu.cn}

\author{Jingtao Ding}
\affiliation{%
  \department{Department of Electronic Engineering}
  \institution{Tsinghua University}
  \city{Beijing}
  \country{China}
}
\email{dingjt15@tsinghua.org.cn}

\author{Yong Li}
\affiliation{%
  \department{Department of Electronic Engineering}
  \institution{Tsinghua University}
  \city{Beijing}
  \country{China}
}
\email{liyong07@tsinghua.edu.cn}

\begin{abstract}
Sequential recommender systems have achieved steady gains in offline accuracy, yet it remains unclear how close current models are to the intrinsic accuracy limit imposed by the data. A reliable, model-agnostic estimate of this ceiling would enable principled difficulty assessment and headroom estimation before costly model development. Existing predictability analyses typically combine entropy estimation with Fano's inequality inversion; however, in recommendation they are hindered by sensitivity to candidate-space specification and distortion from Fano-based scaling in low-predictability regimes. We develop an entropy-induced, training-free approach for quantifying accuracy limits in sequential recommendation, yielding a candidate-size-agnostic estimate. Experiments on controlled synthetic generators and diverse real-world benchmarks show that the estimator tracks oracle-controlled difficulty more faithfully than baselines, remains insensitive to candidate-set size, and achieves high rank consistency with best-achieved offline accuracy across state-of-the-art sequential recommenders (Spearman \(\rho\) up to 0.914). It also supports user-group diagnostics by stratifying users by novelty preference, long-tail exposure, and activity, revealing systematic predictability differences. Furthermore, predictability can guide training data selection: training sets constructed from high-predictability users yield strong downstream performance under reduced data budgets. Overall, the proposed estimator provides a practical reference for assessing attainable accuracy limits, supporting user-group diagnostics, and informing data-centric decisions in sequential recommendation.
\end{abstract}


\maketitle
\pagestyle{plain}

\input{sections/01-introduction}
\input{sections/02-related-work}
\input{sections/03-preliminaries}
\input{sections/04-methodology}
\input{sections/05-experiments}

\section{Conclusion}
This work studies the intrinsic accuracy limits of sequential recommender systems and aims to provide a model-agnostic reference for difficulty assessment and headroom estimation prior to costly model training and deployment. We revisited the widely used entropy-and-Fano paradigm and highlighted its practical challenges in recommendation, including the unavoidable dependence on how the candidate space is defined and the resulting instability under large candidate sets. To address these issues, we proposed an entropy-induced, training-free approach that directly maps an entropy estimate to an attainable-accuracy reference with a clear information-theoretic interpretation, without introducing an explicit candidate-set size parameter.

Extensive experiments on controlled synthetic generators and diverse real-world benchmarks demonstrate the effectiveness and robustness of the proposed approach. On synthetic data with oracle-controlled difficulty, the proposed estimate tracks the difficulty changes more faithfully than representative baselines. Under candidate-size sweeps with a fixed oracle ceiling, it remains stable while Fano-based estimates exhibit pronounced candidate-dependent drift. On real datasets, the resulting dataset-level ordering aligns strongly with the ordering of best-achieved offline accuracy across a broad suite of state-of-the-art sequential recommenders, and analyses by user groups further reveal systematic difficulty differences associated with novelty preference, long-tail exposure, and user activity. We also find that predictability can guide training data selection: training sets constructed from high-predictability users achieve strong downstream performance with less data. These findings suggest that entropy-induced limit characterization can serve as a practical and interpretable tool for understanding why accuracy gains saturate, comparing datasets beyond raw metrics, guiding where further modeling effort is most likely to be effective, and informing data-centric decisions. Future work includes tightening the theoretical characterization, extending the framework to broader evaluation protocols, and integrating the estimate into system-level workflows such as model selection, temporal difficulty monitoring, and data-centric optimization.

\bibliographystyle{ACM-Reference-Format}
\bibliography{references}

\end{document}

%% file: sections/01-introduction.tex
\section{Introduction}
Sequential recommendation models users' time-ordered interaction behaviors, and its canonical formulation is next-item recommendation given a user's historical interaction sequence. This problem arises broadly in e-commerce, content delivery, and short-video platforms, and has attracted sustained research interest~\cite{luo2024collaborative,yao2024recommendertransformers}. Despite continuing advances in model architectures and training paradigms~\cite{zhang2024scalinglaw,cui2024distillation,shi2024diversifying_transformers,zhang2024soft_contrastive_sr,hu2025horae}, two fundamental questions repeatedly arise in both research and practice: (i) whether the given data contain sufficient structural regularities to support accurate prediction, and (ii) under a fixed task specification and evaluation protocol, how far the best-achieved performance is from the attainable accuracy limit, thereby informing whether further modeling and engineering effort is likely to yield meaningful gains. These questions cannot be answered reliably using offline metrics alone, because such metrics summarize the outcome of a \emph{particular model} under a \emph{particular training and evaluation protocol}; their values are jointly influenced by model capacity, optimization choices, and data characteristics, and thus do not directly reveal the intrinsic difficulty or attainable level of the task~\cite{wang2024biasagnostic}. Moreover, attempting to answer these questions by repeatedly training, tuning, and deploying multiple models is costly and time-consuming, and the resulting conclusions are often contingent on the chosen model family and implementation details, lacking a stable point of reference.

To this end, it is necessary to introduce a model-agnostic quantitative concept that characterizes the intrinsic regularity and attainable limit of a task, enabling the assessment of prediction difficulty and potential improvement \emph{prior} to model training and deployment. This concept is \textbf{predictability}: the best achievable predictive performance attainable by any algorithm under the given data and task specification~\cite{song2010limits,jarv2019predictability_session_nextitem}. Quantifying predictability provides a unified and operational basis for judging whether the data are intrinsically predictable and whether there remains room for improvement, and it enables difficulty comparisons across datasets and user populations. More importantly, if predictability can be estimated efficiently at the user level, it may also serve as an upstream signal for data-centric decisions such as selecting more informative training data under limited resource budgets.
\begin{figure*}[!t]
	\centering
	\setlength{\abovecaptionskip}{0pt}
	\setlength{\belowcaptionskip}{0pt}
	\includegraphics[width=0.98\textwidth]{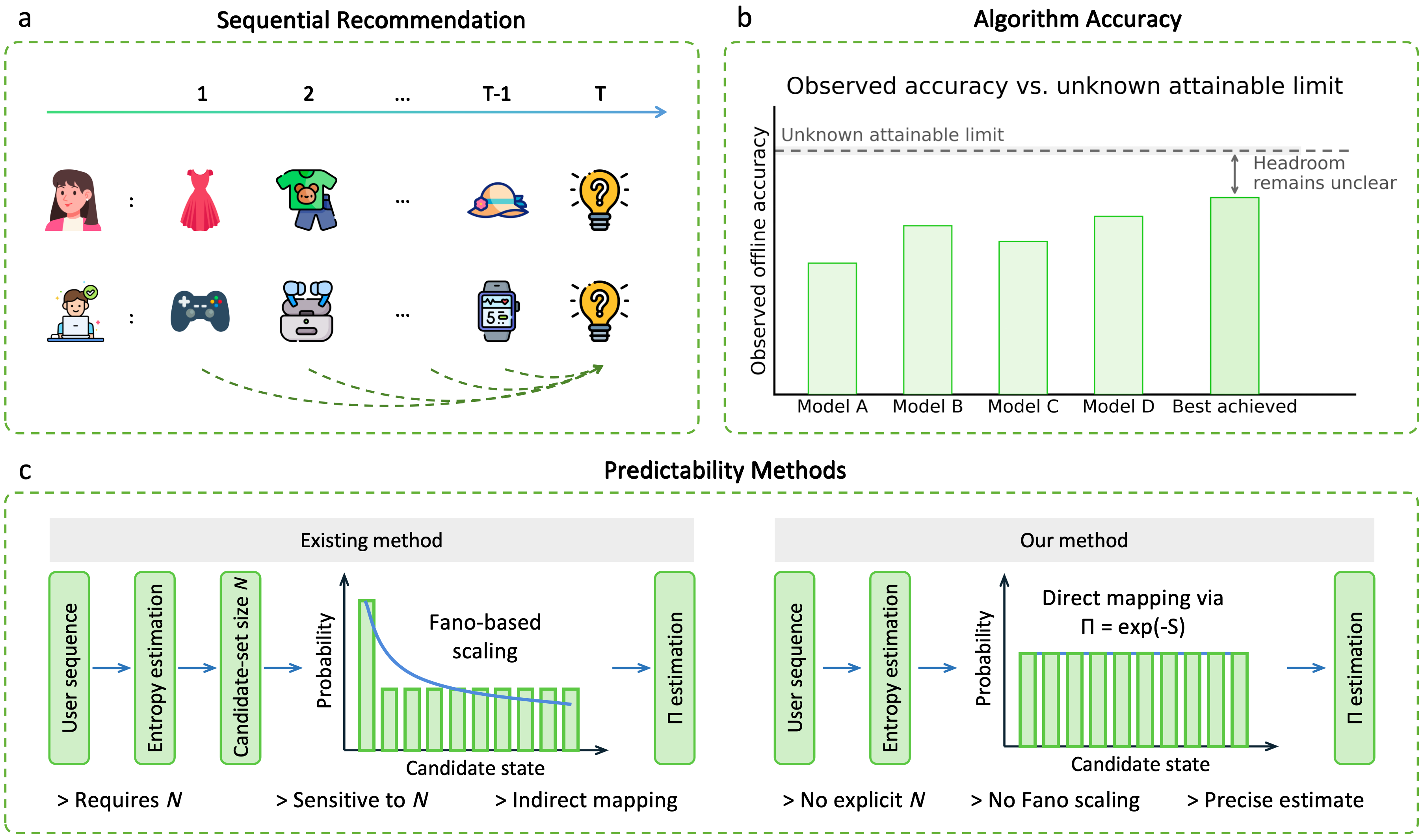} 
	\caption{Overview of accuracy-limit characterization in sequential recommendation. (a) Task illustration. (b) Best-achieved offline accuracy continues to improve, while a model-agnostic attainable reference is still lacking. (c) Entropy-based predictability estimation with Fano scaling versus our entropy-induced characterization without Fano scaling.}
	\Description{Overview of the accuracy-limit characterization problem in sequential recommendation: (a) task illustration, (b) improving best-achieved accuracy across models, and (c) a comparison between entropy-plus-Fano and the proposed entropy-induced characterization.}
	\label{fig:intro_overview}
\end{figure*}

Figure~\ref{fig:intro_overview} summarizes our setting and motivation. Fig.~\ref{fig:intro_overview}(a) illustrates the sequential recommendation task of predicting the next interaction conditioned on a user's historical sequence. Fig.~\ref{fig:intro_overview}(b) shows that the best-achieved offline accuracy of sequential recommenders has continued to improve, while a model-agnostic reference for the attainable accuracy limit remains unclear. To estimate such limits, a widely used line of predictability research follows an ``entropy + Fano scaling'' pipeline: it first estimates sequence entropy as a measure of uncertainty and then maps entropy to predictability by numerically inverting Fano's inequality~\cite{song2010limits,xu2019predictability,sun2020revealing,tang2020predictability}. Fig.~\ref{fig:intro_overview}(c) contrasts this paradigm with our entropy-induced characterization that avoids Fano scaling; as we discuss next, the reliance on Fano scaling introduces key limitations in sequential recommendation.

Despite its broad adoption, applying the ``entropy + Fano scaling'' paradigm to sequential recommendation faces two fundamental obstacles. First, the mapping requires specifying a candidate-set size \(N\), yet in recommendation \(N\) is protocol and system dependent (e.g., the full item universe, a retrieved candidate pool, or sampled evaluation candidates) and thus lacks a single operational definition. This issue is exacerbated by the fact that practical recommender systems typically involve very large item vocabularies: when \(N\) is large, Fano inversion can be strongly re-calibrated by \(N\) and may yield implausible estimates that are artificially pushed toward high predictability, undermining its reliability~\cite{smith2014refined,xu2023quantifying}. Second, Fano scaling is prone to excessive scaling in the low-predictability regime, leading to systematic overestimation~\cite{kulkarni2019examining}. Since sequential recommendation often operates in low-predictability settings, this bias further weakens the method's utility for difficulty assessment and limit characterization. Taken together, these limitations motivate an entropy--predictability transformation that avoids Fano scaling and does not depend on \(N\).

To this end, we propose a predictability characterization that does not rely on Fano scaling. Rather than adopting ``entropy + Fano scaling,'' we start from an information-theoretic interpretation and establish a derivable link between entropy and attainable accuracy. Specifically, we use the entropy of the next-step distribution as a measure of uncertainty and map it to an interpretable \emph{effective} candidate size. Under this view, two prediction problems with the same entropy can be interpreted as lossless identification tasks with the same average information requirement; the corresponding least favorable case is a uniform distribution over candidate states, where no predictor can exploit probability differences, and the attainable accuracy is determined by the reciprocal of the effective size. Building on this idea, we derive an entropy-induced lower bound on predictability and construct a training-free estimator that does not require an explicit candidate-set size parameter. When applied to sequential recommendation, the method only requires estimating entropy from user interaction sequences, enabling a fast characterization of task difficulty and the attainable level without depending on \(N\). The concrete sequence-to-predictability computation procedure is summarized in Algorithm~\ref{alg:entropy_predictability}. This formulation reduces predictability quantification to entropy estimation, thereby enabling model-agnostic limit comparisons and diagnostic analyses in sequential recommendation.

In sequential recommendation, the main practical appeal of our entropy-induced estimator lies in making limit characterization operational under large candidate spaces. Because it reduces predictability estimation to entropy computation on interaction sequences and avoids both Fano scaling and an explicit candidate-set size \(N\), it enables low-cost, model-agnostic assessment of intrinsic difficulty before training and deployment. This further supports diagnostic analyses at multiple granularities: one can stratify users by behavioral dimensions such as novelty preference, long-tail exposure, and activity, and quantify how intrinsic predictability varies across user groups; similarly, tracking predictability over time windows can reveal difficulty drift without repeatedly retraining models. Beyond diagnosis, predictability can also be used as a proactive signal for training data construction: under constrained budgets, prioritizing additional candidate users with high predictability can improve downstream recommendation performance more effectively than unguided selection. Finally, by comparing the dataset-level predictability ranking with the ranking of best-achieved offline accuracy, the estimator can highlight datasets or user groups where current methods appear to underperform relative to intrinsic regularities, providing a principled cue for prioritizing further modeling effort.

Our contributions can be summarized as follows:
\begin{itemize}
	\item We develop an entropy-induced theory for characterizing predictability and derive a lower bound on attainable accuracy without relying on Fano scaling, thereby avoiding the excessive-scaling bias in low-predictability regimes.
	\item We propose a training-free predictability estimator that avoids specifying the candidate-set size \(N\), preventing non-operational definitions and rescaling-induced misestimation in large candidate spaces.
	\item We validate the proposed method on theoretically controlled synthetic data and a diverse set of real-world sequential recommendation benchmarks, demonstrating its utility for stability, robustness to candidate-space scaling, and dataset-level and user-group diagnostic analyses.
	\item We further show that predictability can guide training data selection under fixed evaluation tasks and limited extra-data budgets, extending its role from post-hoc difficulty analysis to practical data construction.
\end{itemize}

%% file: sections/02-related-work.tex
\section{Related Work}
\subsection{Sequential Recommendation}
Sequential recommendation models users' time-ordered interaction histories and aims to predict the next preference conditioned on past behaviors, with next-item recommendation as a canonical task. This problem is prevalent in e-commerce, content delivery, and short-video systems, and serves as a key component connecting users' evolving interests with real-time decision making. Prior work has progressed along two intertwined directions. On the one hand, modeling paradigms for sequential dependence have continuously evolved to better capture short-term intent, long-term preference, and interest drift. On the other hand, task settings and learning paradigms have been expanded, including session-level versus long-term scenarios, multi-source contextual signals, and large-model-based training~\cite{kang2018sasrec,sun2019bert4rec,zhang2024scalinglaw,zou2024knowledge_conversational,sun2024crosscity_poi,wang2025llm4dsr,xin2025llmcdsr,lin2026coderplus,xu2021core_interest_network,xu2023within_basket_aux}. Recent surveys provide systematic overviews and discuss emerging topics such as cross-domain transfer, generative approaches, and large-model-driven recommendation~\cite{pan2026survey_sr,chen2024cdsr_survey}.

Evolution of Modeling Paradigms.
Early sequential recommendation approaches often relied on first-order or low-order Markov assumptions and modeled local dependencies via transition probabilities or similarities. While such formulations are structurally clear and interpretable, they are limited in capturing long-range dependencies and complex interest evolution. To jointly model long-term preference and short-term transitions, subsequent work combined transition structures with personalized preference modeling (e.g., matrix factorization), forming classical hybrid frameworks~\cite{rendle2010fpmc,he2016fossil}. In parallel, metric-space formulations modeled sequential preference as a continuous transition of users in a latent space, providing a compact representation of successive interactions~\cite{he2018transrec}.
With the adoption of deep learning in recommender systems, sequence modeling gradually shifted from explicit transition assumptions to representation learning. Recurrent neural networks (RNN/GRU/LSTM) encode interaction sequences into hidden states, representing time-evolving interests in a state space and strengthening sequential dependence modeling~\cite{hidasi2015gru4rec}. Attention and memory mechanisms were further introduced to improve the selection of informative history segments and aggregation of session intent~\cite{li2017narm,liu2018stamp}. In addition, convolutional and residual architectures have been used to efficiently extract local \(n\)-gram behavior patterns, improving parallelism and complementing sequential modeling~\cite{tang2018caser,yuan2019nextitnet}.

More recently, self-attention explicitly operationalizes the question of which historical interactions are used for a prediction by aggregating weighted representations over past positions. This mechanism is better suited to modeling long-range dependencies, interest drift, and mixtures of interests, and has become a dominant paradigm in Transformer-based sequential recommendation~\cite{kang2018sasrec,sun2019bert4rec,fan2021lightsans,shi2024diversifying_transformers,zhang2025ssd4rec}. Building on this paradigm, many studies have enhanced representation capacity and robustness through architectural and training advances, such as incorporating frequency-domain or filtering perspectives to complement high-frequency or periodic information~\cite{liu2023fearec}, and using contrastive learning and data augmentation to alleviate sparse supervision and representation degeneration~\cite{xie2022cl4srec,qiu2022duorec,zhang2024soft_contrastive_sr}. Furthermore, in session-based sequential recommendation, short-window interaction trajectories often exhibit backtracking, branching, and higher-order co-occurrences; modeling session sequences as directed weighted graphs and applying graph neural networks can capture nonlinear transition structures and higher-order dependencies~\cite{wu2019srgnn,xie2019gcsan,ma2019hgn}. Recent work has also combined graph structures with sequence encoding or contrastive learning to exploit both local sequential context and global structural signals~\cite{zhang2022gcl4sr,zhang2025positional_prompts_sr}.

Extensions in Task Settings, Granularity, and Learning Paradigms.
From a temporal granularity perspective, sequential recommendation is commonly distinguished into session-level and long-term scenarios. Session-level sequences are short and context-dependent, with frequent intent switching; modeling often emphasizes local transitions and immediate preferences. Long-term scenarios involve longer sequences and more pronounced interest drift, requiring mechanisms that capture long-range dependencies while balancing long-term preference and short-term intent. Correspondingly, architectures designed for session intent aggregation and attention-based selection have been developed~\cite{li2017narm,liu2018stamp}. In terms of learning paradigms, sequential recommendation has gradually moved from primarily supervised training to training regimes centered on self-supervision and pretraining, including masked modeling, order recovery, and contrastive learning. With increasing model scale and training data, scalable training and scaling behaviors have also drawn attention, along with strategies such as distillation and structured transfer for large-model-based recommenders~\cite{yao2024recommendertransformers,zhang2024scalinglaw,cui2024distillation}. These advances have contributed to continual improvements in offline metrics, while making the more fundamental question increasingly salient: under a given dataset and task specification, how far current performance is from the attainable limit.

Overall, existing work on sequential recommendation primarily focuses on improving empirical performance through stronger models and training strategies. In contrast, this paper studies model-agnostic characterization of attainable performance, with the goal of providing an interpretable reference for ``accuracy limits'' and task difficulty to support difficulty assessment and headroom analysis prior to model training and deployment.

\subsection{Time-Series Predictability}
Time-series predictability studies aim to answer a fundamental question: under a given dataset and prediction-task specification, what is the best predictive performance that any algorithm can achieve. A recent survey suggests that the literature has formed a relatively clear methodological landscape, in which information-theoretic approaches play a central role by quantifying uncertainty using entropy and mapping uncertainty to an upper bound on attainable performance or to a predictability estimate via inequalities~\cite{xu2026predictability_complex_systems}. This framework typically consists of two components: estimating sequence entropy and establishing the theoretical link that maps entropy to an attainable-performance boundary. Because the true entropy is often intractable to compute, various nonparametric estimators have been proposed, among which Lempel--Ziv-type estimators based on compression principles are representative~\cite{kontoyiannis1998nonparametric}. Building on such entropy estimates, Song \emph{et al.} used human mobility sequences as a canonical case and related entropy to the best attainable accuracy through Fano's inequality, thereby providing an upper-bound estimate of predictability~\cite{song2010limits}. This paradigm has since been widely reused and extended, leading to applications and variants across time-series and complex-system settings~\cite{xu2019predictability,tang2020predictability,sun2020revealing}.

Within this general framework, prior studies have mainly advanced applicability and tightness along two directions. The first line of work focuses on the candidate-space specification in the ``entropy--Fano'' mapping. In its standard form, the Fano framework requires an explicit candidate-set size \(N\), which directly affects numerical tightness and cross-scenario comparability. To mitigate this dependency, existing studies have proposed tightening \(N\) using the number of historical states, reachability constraints, or topological constraints, thereby obtaining tighter upper-bound estimates~\cite{smith2014refined}. Meanwhile, other work has systematically analyzed sources of bias in this paradigm on real data and reported that it may exhibit pronounced overestimation in low-predictability regimes, undermining its use for difficulty assessment~\cite{kulkarni2019examining}. Together, these results indicate that the reliability of ``entropy--Fano'' methods depends not only on entropy-estimation accuracy but also on the candidate-space specification and the looseness introduced by the inequality mapping~\cite{xu2026predictability_complex_systems}.

The second line of work emphasizes \emph{conditional structure}: by incorporating contextual or multi-source information, uncertainty characterization is extended from a single sequence to conditional settings. Beyond Shannon entropy alone, related studies use quantities such as mutual information and conditional entropy to quantify how external information reduces sequence uncertainty, and to analyze the sources and mechanisms of predictability~\cite{teixeira2019deciphering,teixeira2021impact,zhang2022beyond}. In these approaches, conditional entropy characterizes the residual uncertainty given context, allowing predictability analyses to better reflect ``information availability'' in practical systems and improving interpretability~\cite{xu2026predictability_complex_systems}.

In addition to information-theoretic upper-bound approaches, surveys have also summarized complementary complexity measures to characterize regularities in numerical sequences, short sequences, or high-noise settings~\cite{xu2026predictability_complex_systems}. For example, approximate entropy and sample entropy have been used as measures of sequence complexity~\cite{pincus1991approximate,richman2000physiological}, and permutation entropy captures local dynamical structure via ordinal patterns with favorable robustness and computability~\cite{bandt2002permutation}. Moreover, theoretical discussions on the relationship between predictability and optimal error provide an additional perspective for understanding attainable boundaries~\cite{xu2023equivalence}. Building on these lines of work, our study focuses on a direct connection between entropy and predictability to support subsequent method comparisons and analyses in recommendation settings.

Predictability limits in recommender systems have also begun to receive attention. For example, recent work has studied rating prediction and systematically discussed attainable performance bounds and their relationship to data uncertainty~\cite{xu2025upper_bound_rating_prediction}. Compared with rating prediction, however, sequential recommendation typically operates with a much larger candidate space, and the definition of candidate-set size \(N\) is highly protocol dependent, such as full candidates, visible candidates, retrieved candidates, or sampled candidates used during evaluation. This makes Fano-based upper-bound estimation face non-operational choices and potential distortion in recommendation tasks, and has motivated additional constraints or candidate-space tightening strategies tailored to recommendation~\cite{xu2023quantifying}. Therefore, while entropy provides a general tool for characterizing sequence uncertainty, recommendation settings still call for a more operational entropy--predictability transformation that can be used for limit characterization and headroom analysis without depending on a candidate-space size parameter.

%% file: sections/03-preliminaries.tex
\section{Preliminaries}

\subsection{Task Setup and Notation}
We study the next-item prediction task in sequential recommendation. Let \(\mathcal{U}\) denote the set of users and \(\mathcal{I}\) the set of items. The observed data consist of user interaction logs. For each user \(u\in\mathcal{U}\), the time-ordered interaction sequence is denoted as
\begin{equation}
	\mathbf{x}^{(u)} = \left(x^{(u)}_1, x^{(u)}_2, \dots, x^{(u)}_{T_u}\right),\quad x^{(u)}_t \in \mathcal{I},
\end{equation}
where \(T_u\) is the sequence length of user \(u\). For simplicity, we focus on the item sequence in what follows.

Given the history \(h^{(u)}_t\triangleq (x^{(u)}_1,\dots,x^{(u)}_t)\), the goal at time \(t\) is to predict the next interaction \(x^{(u)}_{t+1}\). From a probabilistic perspective, this corresponds to modeling the conditional distribution \(P(X_{t+1}\mid h_t)\). We report empirical results using standard offline ranking metrics in the experimental section.

\subsection{Definition of Predictability}
Predictability characterizes the optimal level that a prediction task can achieve under a given data distribution and task specification. For next-item prediction in sequential recommendation, we view a prediction algorithm as a mapping from the history \(h_t\) to an output. Let \(\mathcal{A}\) denote the set of all candidate algorithms (or prediction rules). For any \(\mathnormal{a}\in\mathcal{A}\), at time \(t\) it outputs a prediction \(\widehat{x}_{t+1}=\mathnormal{a}(h_t)\).

Under a fixed evaluation criterion (e.g., defining a successful prediction as correctly matching the true next item), the expected accuracy of algorithm \(\mathnormal{a}\) is
\begin{equation}
	Accuracy_{\mathnormal{a}} \triangleq \mathbb{E}\!\left[\mathbf{1}\!\left(\mathnormal{a}(h_t)=x_{t+1}\right)\right],
\end{equation}
where the expectation is taken with respect to the data-generating process, i.e., the joint distribution of histories and next items. Accordingly, predictability \(\Pi\) is defined as the optimal accuracy achievable by any algorithm under this task setting:
\begin{equation}
	\Pi \triangleq \sup_{\mathnormal{a}\in\mathcal{A}} Accuracy_{\mathnormal{a}}.
\end{equation}
This definition emphasizes that predictability is a model-agnostic quantity jointly determined by the data distribution and the task specification, and it is intended to characterize the attainable predictive level of the task itself.

\subsection{Entropy Estimation and Fano Mapping}
From an information-theoretic perspective, entropy is a fundamental measure of uncertainty in a stochastic process. In a sequential prediction task, the remaining uncertainty after observing the history \(h_t\) can be characterized by the entropy of the conditional distribution \(P(X_{t+1}\mid h_t)\). In practice, one often considers the average uncertainty across different histories to obtain an overall characterization of regularity in the sequence.

Since the true generating distribution is typically unknown, entropy must be estimated from observed sequences. Common approaches in the predictability literature include nonparametric estimators based on compression principles, among which Lempel--Ziv-type methods are widely used for sequence entropy estimation~\cite{kontoyiannis1998nonparametric}. In our experiments, we adopt this class of nonparametric estimators and denote the resulting estimate of sequence uncertainty by \(\widehat{S}\).

Given an entropy estimate, the classical predictability-quantification route further uses Fano's inequality to establish a mapping from entropy to the optimal attainable accuracy~\cite{song2010limits}. To illustrate its core logic, consider the next-step distribution \(P(X_{t+1}\mid h_t)\) conditioned on a given history \(h_t\). Song \emph{et al.}~\cite{song2010limits} construct an auxiliary distribution \(\hat{P}(X_{t+1}\mid h_t)\) based on the true distribution: among \(N\) candidate states, they keep the probability of the most likely state \(x_{MS}\), denoted \(p_1\triangleq \max_x P(x\mid h_t)\), and rescale the remaining \(N-1\) probabilities to be equal, i.e.,
\begin{equation}
	\hat{P}(X_{t+1}\mid h_t)=\left(p_1,\frac{1-p_1}{N-1},\dots,\frac{1-p_1}{N-1}\right).
\end{equation}
By the maximum-entropy principle, under fixed \(p_1\) and candidate size \(N\), distributing the remaining probability mass uniformly maximizes entropy, and thus
\begin{equation}
	S(X_{t+1}\mid h_t)\le S(\hat{X}_{t+1}\mid h_t).
\end{equation}
Fano's inequality then relates the error probability to conditional entropy. Let \(\widehat{X}_{t+1}=\mathnormal{a}(h_t)\) be the output of an arbitrary prediction rule, and define the error event \(\widehat{X}_{t+1}\neq X_{t+1}\) with probability \(P_e\triangleq \Pr(\widehat{X}_{t+1}\neq X_{t+1})\). For a discrete variable with \(N\) candidate states, Fano's inequality gives
\begin{equation}
	S(X_{t+1}\mid \widehat{X}_{t+1}) \le h_2(P_e) + P_e \log_2 (N-1),
	\label{eqn:fano}
\end{equation}
where \(h_2(\cdot)\) is the binary entropy function. Since \(\widehat{X}_{t+1}=\mathnormal{a}(h_t)\) is a deterministic function of \(h_t\), \(\widehat{X}_{t+1}\) introduces no additional randomness given \(h_t\). Therefore, \(h_t\) contains all information in \(\widehat{X}_{t+1}\) (and \(\widehat{X}_{t+1}\) is a compression of \(h_t\)), and the conditional entropy satisfies
\begin{equation}
	S(X_{t+1}\mid h_t)\le S(X_{t+1}\mid \widehat{X}_{t+1}).
\end{equation}
Defining predictability as the optimal success probability \(\Pi=1-P_e\) and rewriting the inequality in terms of \(\Pi\) yields
\begin{equation}
	S(X_{t+1}\mid h_t) \le -\Pi\log_2\Pi - (1-\Pi)\log_2(1-\Pi) + (1-\Pi)\log_2(N-1).
	\label{eqn:fano_pi}
\end{equation}
Song \emph{et al.} further upper bound \(S(X_{t+1}\mid h_t)\) using the auxiliary distribution \(\hat{P}\) and align this upper bound with the right-hand side of Eq.~(\ref{eqn:fano_pi}), leading to a numerical mapping from entropy to predictability. Given a candidate-set size \(N\)~\cite{song2010limits,smith2014refined,xu2023quantifying}, the entropy value \(S(P)\) can be mapped to predictability \(\Pi\) by numerically solving
\begin{align}
	S(P) = S_F(\Pi) = -\Pi\log_2\Pi - (1-\Pi)\log_2(1-\Pi) + (1-\Pi)\log_2(N-1).
	\label{eqn:song}
\end{align}
Although this framework has been widely used across multiple types of sequential data, it exhibits key limitations in settings with large candidate spaces such as sequential recommendation, which we discuss below under our task setting.

\begin{figure}[t]
	\centering
	\includegraphics[width=0.65\linewidth]{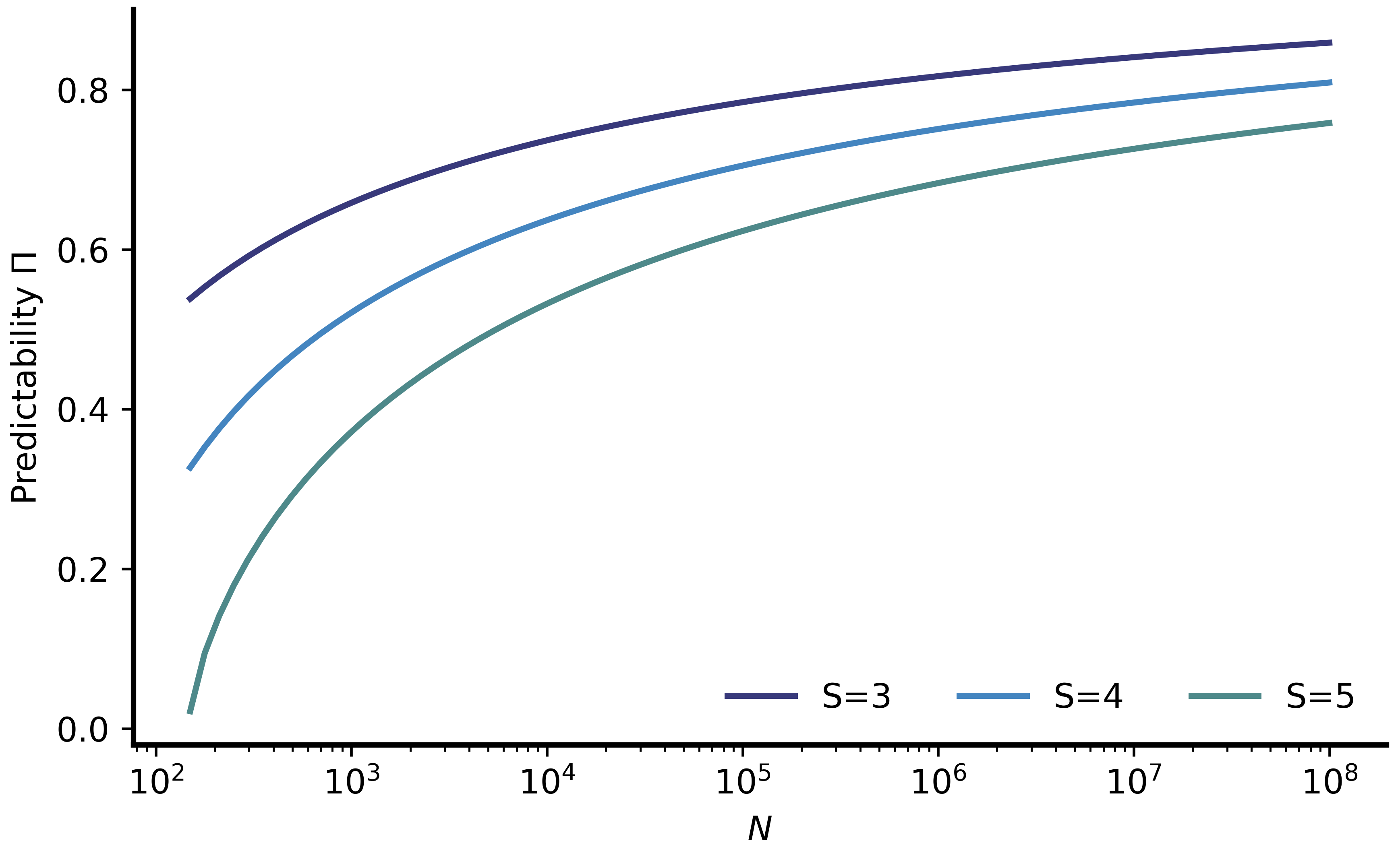}
	\caption{Effect of \(N\) on predictability estimates under the classical Fano mapping. For a fixed entropy level \(S\), the solution \(\Pi\) to Eq.~(\ref{eqn:song}) increases with \(N\), indicating that the choice of \(N\) can substantially change the resulting estimate.}
	\Description{A line chart showing predictability Pi increasing with candidate set size N (log scale) for three fixed entropy values S=3, S=4, and S=5.}
	\label{fig:n_vs_pi_fano}
\end{figure}

\subsection{Limitations of Existing Work}
\noindent Although the ``entropy estimation + Fano mapping'' paradigm has been widely adopted for many types of sequential data, its estimates can become systematically distorted in sequential recommendation, where the candidate space is large. Under our task setting, the limitations mainly manifest in the following two aspects.

\subsubsection{Limitation 1: Excessive Scaling}
To convert an ``entropy constraint'' into a ``constraint on prediction error,'' the classical route typically introduces two upper-bounding steps. First, it upper bounds the true conditional entropy by the entropy of an auxiliary distribution \(\hat{P}\) that is maximal under a fixed maximum probability. Second, it applies Fano's inequality to further translate an entropy upper bound into a constraint on the achievable success rate. The accumulation of these upper bounds makes the numerical mapping loose. In particular, in the low-predictability regime, the resulting estimates are more likely to be biased upward~\cite{kulkarni2019examining}. Consequently, when a sequence exhibits weak regularities and the attainable success rate should be low, ``entropy estimation + Fano mapping'' can yield an overly optimistic predictability level, reducing its discriminative power as a tool for quantifying task difficulty.

\subsubsection{Limitation 2: Dependence on Candidate-Set Size}
Eq.~(\ref{eqn:song}) shows that the classical mapping depends explicitly on the candidate-set size \(N\). In sequential recommendation, however, the ``candidate set'' is not a task-independent constant. On the one hand, the global item vocabulary can be extremely large. On the other hand, offline evaluation often introduces protocol-dependent constraints, such as sampled candidates, filtering rules, or a retrieval stage, leading to substantial variation in the effective candidate space across different settings. Because the term \((1-\Pi)\log_2(N-1)\) is highly sensitive to \(N\), ambiguity in how \(N\) is quantified can be amplified into the predictability estimate. When \(N\) is large, the mapping can push the estimate toward values close to \(1\), which is not aligned with the intended semantics that predictability is a model-agnostic limit determined by the regularity of the data. Under evaluation protocols with sampled candidates, the manual choice of \(N\) further makes estimates difficult to compare across datasets or experimental settings. Fig.~\ref{fig:n_vs_pi_fano} illustrates this issue by showing that, for a fixed entropy value \(S\), the predictability \(\Pi\) implied by Eq.~(\ref{eqn:song}) increases monotonically with \(N\), highlighting the potential distortion in large-scale item spaces.

Taken together, these limitations suggest that, in sequential recommendation, the classical ``entropy + Fano'' paradigm may not provide stable and interpretable predictability estimates. This motivates the entropy--predictability quantification method in our methodology section, which avoids Fano scaling and does not require an explicit candidate-set size \(N\).

%% file: sections/04-methodology.tex
\section{Methodology}
\subsection{A Direct Entropy-to-Predictability Mapping}\label{section2}
Let \(P(X_n\mid h_{n-1})\) denote the conditional distribution of the next state given the history \(h_{n-1}\). Its entropy is defined as
\begin{equation}
	S\!\left(P(\cdot\mid h_{n-1})\right)\triangleq -\sum_{x} P(x\mid h_{n-1})\log P(x\mid h_{n-1}).
\end{equation}
Throughout this paper, \(\log\) denotes the natural logarithm; if entropy is measured in bits with base \(2\), \(\exp(\cdot)\) should be replaced by \(2^{(\cdot)}\) accordingly.
For a discrete variable supported on \(N\) states, \(S\!\left(P(\cdot\mid h_{n-1})\right)\le \log N\), with equality if and only if the distribution is uniform. Motivated by this fact, we define
\begin{equation}
	M(h_{n-1}) \triangleq \exp\!\big(S(P(\cdot\mid h_{n-1}))\big)
\end{equation}
and interpret it as an entropy-induced \emph{effective uncertainty}: under a given entropy value, \(M(h_{n-1})\) corresponds to the ``effective candidate size'' of a uniform distribution with the same entropy. This quantity satisfies \(1\le M(h_{n-1})\le N\) and is completely determined by entropy, thus requiring no explicit specification of the candidate-set size \(N\).

To map entropy to a predictability characterization, we introduce a standard reference. Let \(U_m\) be the uniform distribution over \(m\) candidate states. Then \(S(U_m)=\log m\), and the optimal prediction accuracy under this reference is \(1/m\). Accordingly, for a general conditional distribution \(P(\cdot\mid h_{n-1})\), we associate its uncertainty with the effective size \(m=M(h_{n-1})=\exp(S(P(\cdot\mid h_{n-1})))\), and define the entropy-induced predictability measure
\begin{equation}
	\Pi_S(h_{n-1})\triangleq \frac{1}{M(h_{n-1})}=\exp\!\big(-S(P(\cdot\mid h_{n-1}))\big).
	\label{eq:pi_s_def}
\end{equation}
This measure depends only on entropy and can therefore quantify the intrinsic difficulty of sequential prediction without training a recommendation model. The next subsection presents a formal theoretical result relating \(\Pi_S\) to the predictability of the original task, together with its information-theoretic interpretation and its properties regarding candidate-size robustness and computational efficiency.

\subsection{Entropy-Induced Lower Bound and Information-Theoretic Interpretation}
This subsection presents our core theoretical result: the entropy of the next-step distribution conditioned on a given history directly implies a rigorous lower bound on the predictability of the original task. We then provide an operational interpretation based on lossless coding to clarify the meaning of \(M(h_t)=\exp(S)\): it corresponds to the size of a uniform candidate set (perplexity) that carries the same amount of information as the entropy, and \(\Pi_S(h_t)=1/M(h_t)\) is the corresponding reference-level hit rate under this uniform benchmark.

Given a history \(h_t\), consider the conditional distribution of the next state \(P(X_{t+1}\mid h_t)\). Under the criterion that a prediction is correct if it hits the true next state, predictability is defined as
\begin{equation}
	\Pi^{\ast}(h_t)\triangleq \sup_{\mathnormal{a}} \Pr\!\left(\mathnormal{a}(h_t)=X_{t+1}\mid h_t\right)=\max_x P(x\mid h_t).
	\label{eq:pi_star_def}
\end{equation}
The entropy of the next-step distribution conditioned on \(h_t\) is \(S(P(\cdot\mid h_t))\triangleq -\sum_x P(x\mid h_t)\log P(x\mid h_t)\), and the entropy-induced measure \(\Pi_S(h_t)\) is given in Eq.~(\ref{eq:pi_s_def}).

\begin{theorem}[Entropy-induced lower bound]\label{thm:entropy_lower_bound}
For any history \(h_t\), the predictability of the original task satisfies
\begin{equation}
	\Pi^{\ast}(h_t)\ \ge\ \Pi_S(h_t).
	\label{eq:pi_lower_bound}
\end{equation}
\end{theorem}

\begin{proof}
Let \(p_{\max}(h_t)\triangleq \max_x P(x\mid h_t)=\Pi^{\ast}(h_t)\). For any \(x\), we have \(P(x\mid h_t)\le p_{\max}(h_t)\), and hence \(-\log P(x\mid h_t)\ge -\log p_{\max}(h_t)\). Taking the expectation with respect to \(P(x\mid h_t)\) yields
\begin{align}
	S(P(\cdot\mid h_t))
	&= -\sum_x P(x\mid h_t)\log P(x\mid h_t)\nonumber\\
	&\ge -\sum_x P(x\mid h_t)\log p_{\max}(h_t)\nonumber\\
	&= -\log p_{\max}(h_t),
	\label{eq:entropy_min_entropy}
\end{align}
Exponentiating both sides of Eq.~(\ref{eq:entropy_min_entropy}) and substituting \(p_{\max}(h_t)=\Pi^{\ast}(h_t)\) gives
\begin{equation}
	\Pi^{\ast}(h_t)=p_{\max}(h_t)\ge \exp\!\big(-S(P(\cdot\mid h_t))\big)=\Pi_S(h_t).
\end{equation}
\end{proof}

\begin{figure*}[!t]
	\centering
	\setlength{\abovecaptionskip}{0pt}
	\setlength{\belowcaptionskip}{0pt}
	\includegraphics[width=\textwidth]{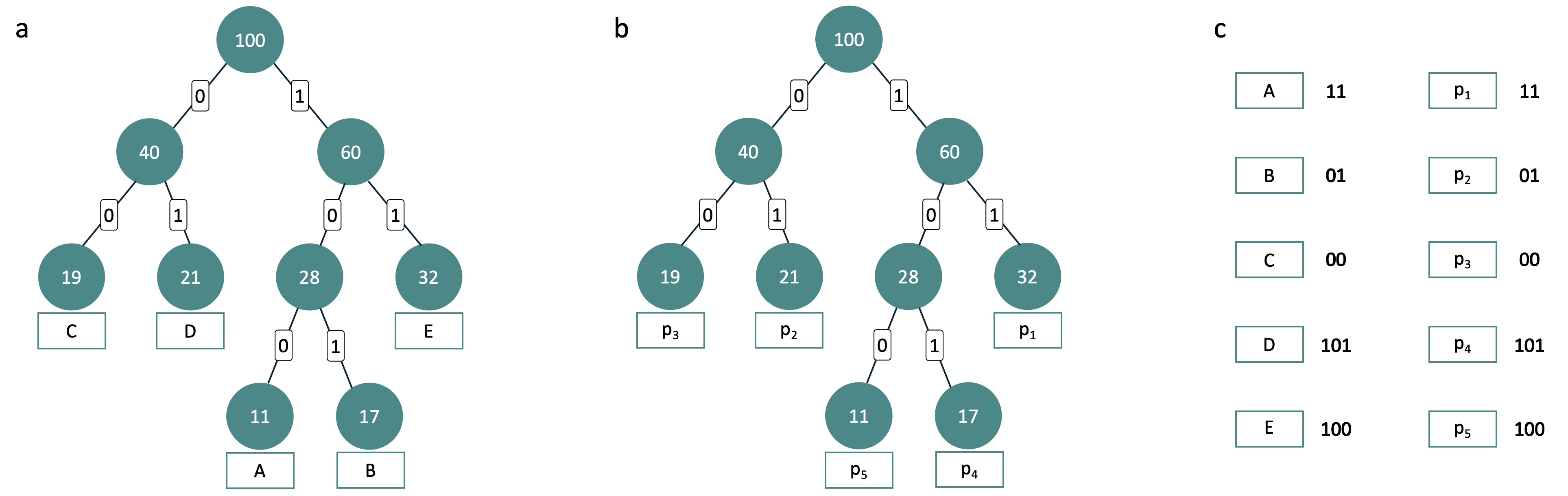}
	\caption{Entropy-induced predictability under a lossless Huffman-coding perspective. (a) Given a history \(h_t\), Huffman coding assigns codewords to next-state symbols according to the conditional distribution \(P(X_{t+1}\mid h_t)\). (b) An equivalent view: the probability mass function \(\{p(x\mid h_t)\}\) determines symbol code lengths \(\ell(x)\), so that ``locating the next state on average'' is operationally equivalent to ``the expected code length.'' Since the optimal expected code length matches entropy up to an additive constant, this perspective motivates the effective size \(M(h_t)=\exp(S)\) and the reference-level hit rate \(\Pi_S(h_t)=1/M(h_t)\).}
	\Description{Entropy-induced predictability under a lossless Huffman-coding view: (a) Huffman coding assigns variable-length codes to next-state symbols according to the conditional distribution; (b) equivalently, the probability mass function induces code lengths whose optimal expected length matches entropy up to a constant, motivating M=exp(S) and Pi_S=1/M.}
	\label{fig:entropy_mapping_coding}
\end{figure*}

\begin{remark}[Operational meaning of \(M(h_t)\) via lossless coding]\label{rem:coding_operational_meaning}
As illustrated in Fig.~\ref{fig:entropy_mapping_coding}, after observing a history \(h_t\), \textit{the uncertainty of the next state \(X_{t+1}\) can be understood not only through entropy but also from the perspective of the cost of lossless identification.} Consider any binary prefix-free code \(c:\mathcal{X}\rightarrow\{0,1\}^{\ast}\), with code length \(\ell(x)\triangleq |c(x)|\). Because such codes admit instantaneous unique decoding, ``losslessly identifying \(X_{t+1}\)'' is operationally equivalent to ``losslessly transmitting its codeword.'' Accordingly, the average identification cost is naturally measured by the expected code length
\begin{equation}
	\bar{L}(h_t)\triangleq \mathbb{E}[\ell(X_{t+1})\mid h_t]=\sum_x p(x\mid h_t)\,\ell(x)
\end{equation}
The theory of lossless prefix coding further implies that the conditional entropy \(H_2(X_{t+1}\mid h_t)\) (with \(\log_2\)) lower bounds the expected code length of any prefix-free code, and that Huffman coding~\cite{huffman1952method} achieves this bound with redundancy strictly less than \(1\) bit:
\begin{equation}
	H_2(X_{t+1}\mid h_t)\ \le\ \bar{L}_{\text{Huff}}(h_t)\ <\ H_2(X_{t+1}\mid h_t)+1.
\end{equation}
Therefore, entropy can be interpreted as the minimum average information required to losslessly identify the next state, up to a fixed additive constant. When using natural logarithms (nats), \(S(P(\cdot\mid h_t))=(\ln 2)\,H_2(X_{t+1}\mid h_t)\), and defining
\begin{equation}
	M(h_t)\triangleq \exp\!\big(S(P(\cdot\mid h_t))\big)
\end{equation}
converts this information quantity into an equivalent size. Specifically, \(M(h_t)\) corresponds to the size of a uniform candidate set (perplexity) with the same information content as the entropy. Under this uniform reference, the hit rate is \(1/M(h_t)\), which is consistent with \(\Pi_S(h_t)=1/M(h_t)\) in Eq.~(\ref{eq:pi_s_def}).

\end{remark}

\subsection{Implementation of the Estimator}
The previous two subsections establish the entropy-induced predictability lower bound \(\Pi_S\) and its theoretical properties at the distribution level (see Eq.~(\ref{eq:pi_s_def}) and Theorem~\ref{thm:entropy_lower_bound}). In practice, however, we only observe finite-length discrete interaction sequences and cannot directly access the true conditional distribution or its entropy. We therefore follow an ``estimate entropy and apply the mapping'' procedure: we first compute an entropy estimate \(\hat{S}\) from the observed sequence, and then obtain the EPL predictability estimate via the closed-form mapping \(\Pi_{\mathrm{EPL}}\triangleq \exp(-\hat{S})\). Below we use sample entropy as an illustrative entropy estimator; in experiments, other nonparametric entropy estimators can be substituted without changing the overall framework. For a direct overview of the complete computational pipeline, see \textbf{Algorithm~\ref{alg:entropy_predictability}}.

Given a discrete sequence \(\mathbf{x}=(x_1,x_2,\dots,x_T)\), we first use sample entropy (SampEn) to estimate its uncertainty level~\cite{richman2000physiological}. Intuitively, SampEn quantifies how the probability of pattern matches decays when the pattern length increases from \(m\) to \(m+1\): more regular sequences exhibit more persistent patterns and therefore lower SampEn, whereas less regular sequences yield higher SampEn. In implementation, we treat \(\textsc{SampEn}(\mathbf{x};m,r)\) as a black-box estimator that takes the sequence and hyperparameters \(m\) (embedding dimension) and \(r\) (tolerance threshold) as input, and outputs an entropy estimate \(\hat{S}\) measured in nats.

We then apply the entropy--predictability mapping to obtain the EPL predictability estimate:
\begin{equation}
	\Pi_{\mathrm{EPL}}\triangleq \exp(-\hat{S}).
	\label{eq:pi_hat_alg1}
\end{equation}
Algorithm~\ref{alg:entropy_predictability} summarizes the full procedure from sequence input to predictability output.

\begin{algorithm}[t]
\caption{EPL: Entropy-based Predictability Lower Bound Estimator}
\label{alg:entropy_predictability}
\begin{algorithmic}[1]
\Require A sequence $\mathbf{x}=(x_1,x_2,\dots,x_T)$, $x_t\in\mathcal{I}$;
sample entropy parameters $m$ and $r$
\Ensure Predictability estimate $\Pi_{\mathrm{EPL}}\in(0,1]$
\State $\hat{S}\gets \textsc{SampEn}(\mathbf{x};m,r)$
\State $\Pi_{\mathrm{EPL}}\gets \exp(-\hat{S})$
\State \Return $\Pi_{\mathrm{EPL}}$
\end{algorithmic}
\end{algorithm}

By Theorem~\ref{thm:entropy_lower_bound}, \(\Pi_S(h_t)=\exp(-S(P(\cdot\mid h_t)))\) is a per-history strict lower bound of the original-task predictability \(\Pi^{\ast}(h_t)\). This bound is determined solely by the data-generating distribution and is independent of any particular learning algorithm, thereby enabling a consistent characterization of the intrinsic difficulty of sequential prediction without training a model.
Moreover, this lower bound depends only on entropy and does not require defining or estimating an impractical candidate-set size \(N\) in sequential recommendation. This property avoids the high sensitivity of classical mappings to \(N\) and mitigates distortions induced by differences in evaluation protocols (e.g., candidate sampling and filtering rules), making \(\Pi_S\) more suitable as a unified reference across datasets and protocols.
From a computational perspective, the cost of \(\Pi_S\) is dominated by entropy estimation; once an entropy estimate is available, the predictability lower bound follows immediately from the closed-form mapping. This leads to good scalability on large collections of user sequences and supports comparisons across user groups or time windows for predictability analysis and drift tracking.
Finally, compared with the classical route that maps entropy to an upper-bound-type estimate via inequality-based scaling, \(\Pi_S\) provides an \(N\)-free lower-bound reference. The two perspectives are complementary in interpretation: the former aims to obtain an upper-bound characterization through a scaling chain, whereas the latter provides a baseline level directly induced by entropy. In sequential recommendation, the parameter-free form of \(\Pi_S\) makes it particularly convenient for comparing results across different datasets and evaluation protocols.

%% file: sections/05-experiments.tex
\section{Experiments}

\subsection{Datasets}\label{sec:exp_datasets}
We conduct empirical evaluations on a set of publicly available benchmark datasets that are widely used in recommender-system research. These datasets span diverse application domains, including movie ratings, music consumption, social bookmarking, e-commerce transactions, and educational interactions. To keep the task setting consistent across datasets, our real-data experiments only use the interaction triples \((\text{user},\text{item},\text{timestamp})\) to form time-ordered sequences.

\begin{itemize}
		\item \textit{AOTM}. A playlist/playlist-sequence dataset with explicitly ordered interactions, commonly used for benchmarking session-based and sequential recommendation~\cite{mcfee2012hypergraph}.
		\item \textit{Delicious}. User interaction logs from a social-bookmarking platform; we use the public HetRec 2011 release~\cite{hetrec2011}.
		\item \textit{Online Retail}. Online Retail II contains real-world online retail transaction logs with timestamped purchase events and product information, and has been used for sequential prediction and behavioral analysis in e-commerce~\cite{chen2012onlineretail2}.
		\item \textit{Personality}. A dataset released by GroupLens; in this paper, we use only its timestamped interaction sequences~\cite{nguyen2017personality}.
		\item \textit{LastFM}. Listening logs from a music service that naturally form time-ordered implicit-feedback sequences; we use the public HetRec 2011 release~\cite{hetrec2011}.
		\item \textit{TaFeng}. A grocery transaction dataset that records customers' purchase events over time, widely used for modeling shopping sequences and next-purchase prediction~\cite{rendle2010fpmc}.
		\item \textit{MovieLens-100K}. A classic movie-rating benchmark with explicit ratings and timestamps, extensively used for benchmarking collaborative filtering and recommendation models~\cite{harper2015movielens}.
		\item \textit{MovieLens-1M}. A larger MovieLens variant with approximately one million ratings, enabling evaluation under higher data volume and increased sparsity relative to MovieLens-100K~\cite{harper2015movielens}.
		\item \textit{MovieLens-20M}. A large-scale MovieLens variant with tens of millions of ratings, commonly used for large-scale collaborative filtering and temporal recommendation research~\cite{harper2015movielens}.
		\item \textit{Algebra}. An educational-interaction dataset from the KDD Cup 2010 educational recommendation task, containing longitudinal interactions between learners and learning items~\cite{stamper2016kddcup2010}.
		\item \textit{Bridge}. Another dataset from the KDD Cup 2010 educational recommendation task, with interaction records analogous to Algebra but from a different setting and content collection~\cite{stamper2016kddcup2010}.
\end{itemize}

Table~\ref{tab:realdata_stats} summarizes the dataset domain and key statistics---the numbers of users, items, and interactions, as well as the average per-user sequence length---to characterize differences in scale, sparsity, and sequential structure, and to provide necessary context for cross-dataset comparisons in subsequent analyses.

\begin{table}[t]
	\caption{Overview of real-world datasets: domain and key statistics.}
	\label{tab:realdata_stats}
	\centering
	\begin{tabular}{llcccc}
		\toprule
		Dataset & Domain & \#Users & \#Items & \#Interactions & Avg. len.\\
		\midrule
		AOTM & Music playlist & 102 & 1{,}962 & 2{,}000 & 19.61\\
		Delicious & Social bookmarking & 170 & 10{,}345 & 11{,}087 & 65.22\\
		Online Retail & E-commerce transaction & 1{,}295 & 2{,}979 & 50{,}000 & 38.61\\
		Personality & User study & 44 & 16{,}408 & 50{,}000 & 1{,}136.36\\
		LastFM & Music listening & 1{,}892 & 12{,}523 & 87{,}061 & 46.02\\
		TaFeng & Grocery transaction & 32{,}266 & 23{,}812 & 743{,}228 & 23.03\\
		MovieLens-100K & Movie rating & 943 & 1{,}682 & 100{,}000 & 106.04\\
		MovieLens-1M & Movie rating & 6{,}040 & 3{,}706 & 1{,}000{,}209 & 165.60\\
		MovieLens-20M & Movie rating & 121{,}191 & 25{,}732 & 17{,}519{,}484 & 144.56\\
		Algebra & Education interaction & 575 & 211{,}397 & 813{,}661 & 1{,}415.06\\
		Bridge & Education interaction & 1{,}146 & 208{,}232 & 3{,}686{,}871 & 3{,}217.16\\
		\bottomrule
	\end{tabular}
\end{table}

\subsection{Recommendation Algorithms}\label{sec:exp_algorithms}
To characterize the attainable performance level of each real-world dataset under a standard offline evaluation protocol, we train and tune a representative set of sequential recommendation algorithms and use the best test performance as a reference for the achievable offline hit rate on that dataset. The selected models cover classical Markov-factorization approaches, RNN/CNN-based sequence models, self-attention-based Transformer architectures, and graph-neural and attention mechanisms designed for session-based recommendation. Below we briefly summarize the core modeling idea of each algorithm.

\begin{itemize}
		\item \textit{SASRec}~\cite{kang2018sasrec}. A Transformer-based sequential model with causal self-attention, which aggregates historical interactions via masked attention to capture long-range dependencies under a parallel computation framework.
		\item \textit{GRU4Rec}~\cite{hidasi2015gru4rec}. A GRU-based session recommender that encodes the interaction sequence into hidden states for next-step prediction, often optimized with ranking-oriented objectives in offline evaluation.
		\item \textit{BERT4Rec}~\cite{sun2019bert4rec}. A bidirectional Transformer trained with a masked-item prediction objective, leveraging richer contextual information to learn sequential representations.
		\item \textit{Caser}~\cite{tang2018caser}. A convolutional sequence-embedding model that uses one-dimensional convolutions on the interaction matrix to extract local patterns and compositional features, capturing short-term dependencies with improved training parallelism.
		\item \textit{NextItNet}~\cite{yuan2019nextitnet}. A sequence model based on dilated convolutions and residual blocks, which expands receptive fields to model longer histories without recurrent computation.
		\item \textit{FPMC}~\cite{rendle2010fpmc}. A factorized personalized Markov chain model that unifies long-term preferences (matrix factorization) and short-term transitions (Markov chains) in a single factorization framework.
		\item \textit{FOSSIL}~\cite{he2016fossil}. An extension of FPMC that further integrates similarity-based components with Markov transitions to improve robustness under sparse interactions and to capture short-term order signals more stably.
		\item \textit{TransRec}~\cite{he2018transrec}. A translation-based model that represents sequential transitions as translation operations in the embedding space, providing a concise geometric characterization of adjacent interactions.
		\item \textit{NARM}~\cite{li2017narm}. A neural attentive session recommender that encodes sessions with an RNN and applies attention over hidden states to emphasize historical fragments relevant to the current intent.
		\item \textit{STAMP}~\cite{liu2018stamp}. A session-based model that emphasizes short-term interests by combining attention with a memory-prioritization mechanism, highlighting the contribution of recent interactions under fast intent shifts.
		\item \textit{SRGNN}~\cite{wu2019srgnn}. A session-based approach that converts a session sequence into a directed graph and applies graph neural propagation to model nonlinear transitions and higher-order co-occurrences.
		\item \textit{GCSAN}~\cite{xie2019gcsan}. A hybrid model that combines graph-based representations with self-attention, learning session-graph dependencies while explicitly modeling the contribution of key positions in the sequence.
		\item \textit{HGN}~\cite{ma2019hgn}. A hierarchical gating model that fuses sequential signals at different granularities to dynamically balance short-term intent and long-term preference.
		\item \textit{RepeatNet}~\cite{ren2019repeatnet}. A session-based model designed for repeat-consumption behavior, introducing a gated mechanism to switch between repeat and exploration modes.
		\item \textit{SINE}~\cite{tan2021sine}. A sequential modeling framework with information-injection and representation-enhancement mechanisms to improve representation quality and generalization under sparse interactions.
		\item \textit{SHAN}~\cite{ying2018shan}. A hierarchical attention model that jointly captures short-term intent and long-term preference via layered attention over within-session interactions and global history.
		\item \textit{LightSANs}~\cite{fan2021lightsans}. A lightweight self-attention model for session recommendation that simplifies attention computation to reduce complexity while maintaining expressive capacity.
		\item \textit{FEARec}~\cite{liu2023fearec}. A frequency-enhanced hybrid-attention sequential model that incorporates frequency-domain modeling components to capture periodic patterns and multi-scale dependencies.
\end{itemize}

\subsection{Predictability Baselines}\label{sec:exp_predictability_baselines}
\noindent To compare against the proposed entropy-induced predictability estimator, we consider three representative predictability baselines. The first is the Fano-scaling framework, which has been widely used in the predictability literature~\cite{song2010limits,smith2014refined}; we instantiate it using either a global candidate size or an empirical reachability-based size. The second is a predictability mapping derived from permutation entropy, which measures local ordinal-pattern complexity of time series~\cite{scarpino2019outbreaks}. The Fano framework has also been applied to estimating predictability limits in session-based next-item recommendation~\cite{jarv2019predictability_session_nextitem}.
\begin{itemize}
	\item \textit{Fano predictability (\(\Pi_{\mathrm{Fano}}\)).}
	This baseline uses Fano scaling to map entropy to predictability~\cite{song2010limits,smith2014refined}: given an entropy estimate \(S\) and a candidate-set size \(N\), it numerically inverts Eq.~(\ref{eqn:song}) to obtain \(\Pi_{\mathrm{Fano}}\). In real-data experiments, we set \(N\) to the global item vocabulary size, i.e., \(N=|\mathcal{I}|\).

	\item \textit{Reachability-Fano predictability (\(\Pi_{\mathrm{Fano},N_r}\)).}
	Since the global candidate-set size may not reflect the effective branching complexity of one-step prediction in real sequential data~\cite{smith2014refined}, this baseline replaces the global \(N\) with an empirical reachability-based size. Specifically, viewing the sequence as a first-order transition process, for each current state \(x\) we collect the set of next states observed in the data, denoted \(\mathcal{N}(x)\), and define
	\begin{equation}
		N_r \;\triangleq\; \max_x |\mathcal{N}(x)|.
		\label{eq:reachability_nr}
	\end{equation}
	Intuitively, \(N_r\) quantifies the maximum number of distinct next choices from the most complex current state. We then replace \(N\) in Eq.~(\ref{eqn:song}) with \(N_r\) and solve the same equation numerically to obtain \(\Pi_{\mathrm{Fano},N_r}\). This baseline can provide a closer characterization of the local candidate space in large-scale sparse data, but its definition relies on a first-order approximation and finite-sample transition statistics.

	\item \textit{Permutation predictability (\(\Pi_{\mathrm{perm}}\)).}
	This baseline follows the use of permutation entropy for quantifying predictability in disease outbreak forecasting~\cite{scarpino2019outbreaks}. Permutation entropy measures structural complexity by counting relative order patterns within local windows~\cite{bandt2002permutation} and can be evaluated across multiple window scales. Given embedding dimension \(d\) and time delay \(\tau\), we construct length-\(d\) embedding vectors and convert each into an ordinal pattern by sorting their values. The empirical frequencies of ordinal patterns yield a Shannon entropy \(H(d,\tau)\), which is normalized by \(\log(d!)\) as
	\begin{equation}
		S_{\mathrm{perm}}(d,\tau)=\frac{H(d,\tau)}{\log(d!)}\in[0,1].
		\label{eq:perm_entropy}
	\end{equation}
	We take the minimum complexity over a prespecified set of scales, \(S_{\mathrm{perm}}=\min_{d,\tau} S_{\mathrm{perm}}(d,\tau)\), and map it to predictability as
	\begin{equation}
		\Pi_{\mathrm{perm}} \;=\; 1 - S_{\mathrm{perm}}.
		\label{eq:pi_perm}
	\end{equation}
	In this paper we use \(\tau=1\) and take the minimum over \(d\in\{3,4,5\}\). This baseline does not depend on a global candidate-set size; however, it measures the complexity of local ordinal structures, which is not fully aligned with the semantics of next-item classification.
\end{itemize}

\begin{figure}[t]
	\centering
	\begin{minipage}{0.49\linewidth}
		\centering
		\includegraphics[width=\linewidth]{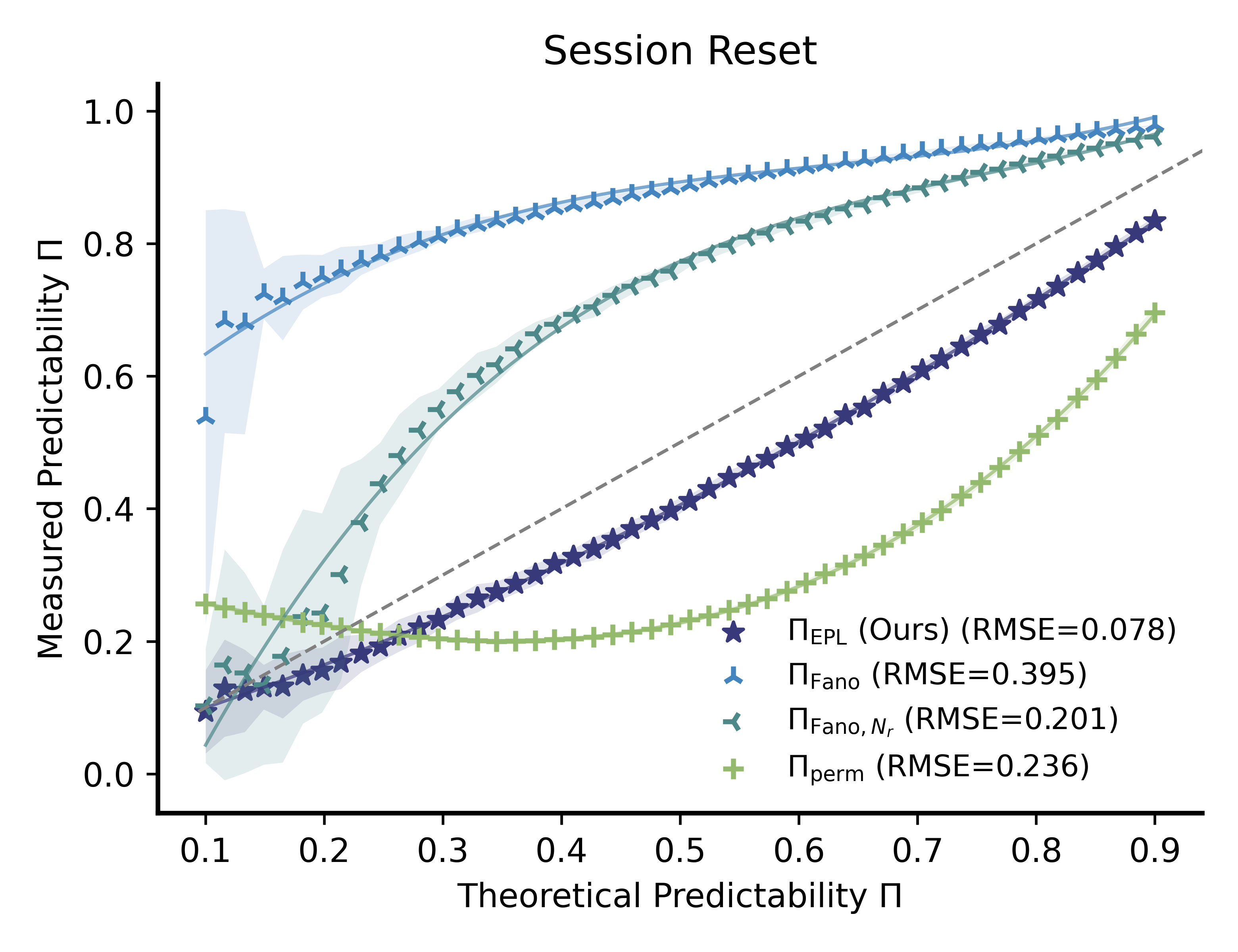}
	\end{minipage}\hfill
	\begin{minipage}{0.49\linewidth}
		\centering
		\includegraphics[width=\linewidth]{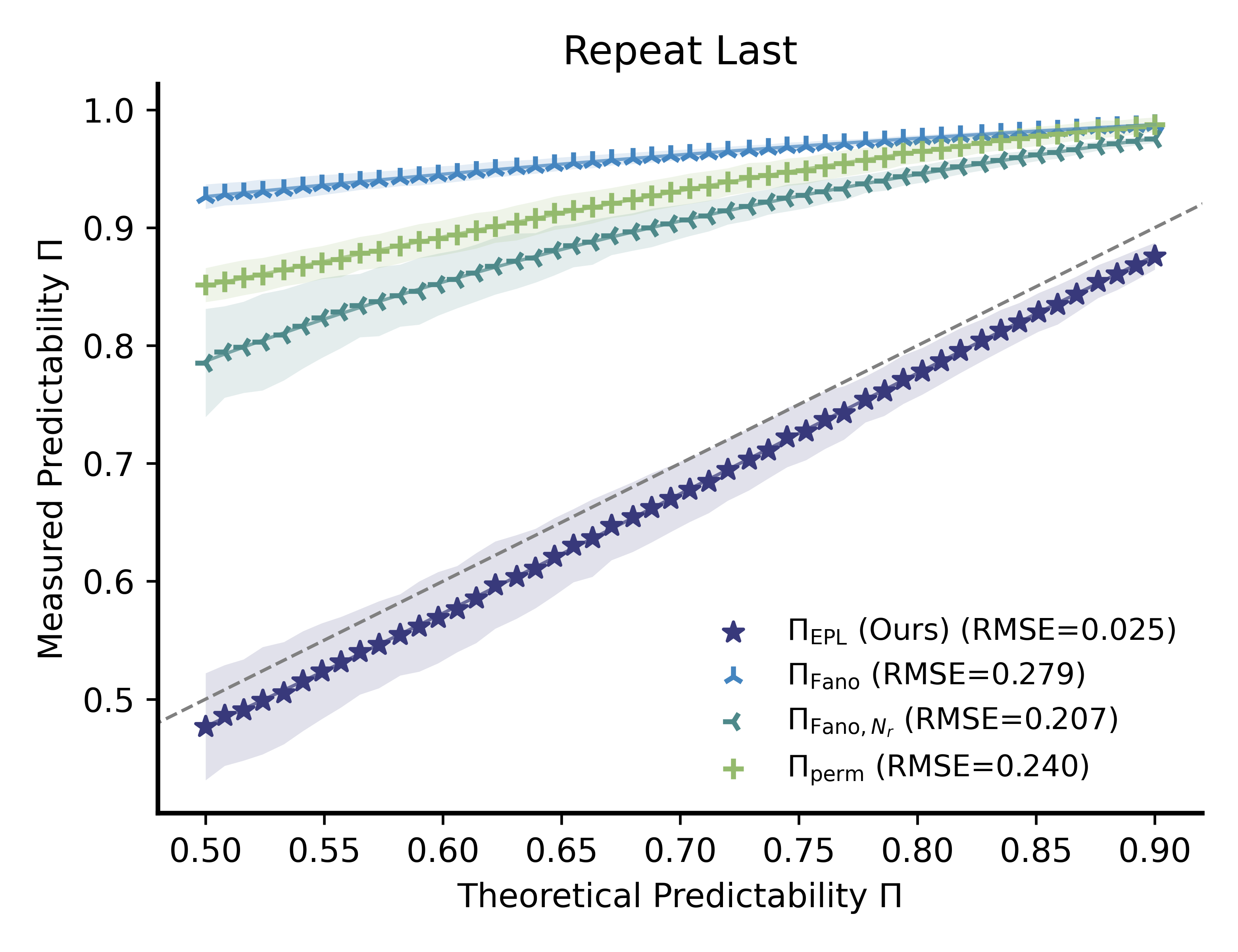}
	\end{minipage}
	\caption{Predictability estimation on two synthetic generators: \emph{Session Reset} and \emph{Repeat Last}. The horizontal axis shows the theoretical predictability, and the vertical axis reports predictability estimates produced by different methods.}
	\Description{Two-panel comparison of predictability estimates against an ideal reference upper bound on Session-Reset and Repeat-Last synthetic generators.}
	\label{fig:synth_two_generators}
\end{figure}

\subsection{Controlled Validation on Synthetic Sequences}\label{sec:exp_synth}
To compare different predictability estimators in a fully controlled setting, we construct two synthetic sequence generators: \emph{Session Reset} and \emph{Repeat Last}. Both generators produce user--item interaction sequences and share a key property: conditional on the latent state of the generator (e.g., the current preference set or the previous interaction), the next-step distribution admits a closed form, which allows an analytic derivation of the optimal hit rate under an ideal predictor with access to the latent state. We denote this ideal reference ceiling by \(\mathrm{Hit@1}^{\mathrm{Oracle}}\), where ``Oracle'' indicates that the predictor is allowed to observe the latent state at each time step. In experiments, we use \(\mathrm{Hit@1}^{\mathrm{Oracle}}\) to align task difficulty, and generate synthetic datasets spanning different difficulty levels by inverting the corresponding noise parameter. For each difficulty point, we repeat the generation multiple times and report the mean and error bars of predictability estimates.

\paragraph{Session Reset generator.}
This process captures a dynamics where the preference set occasionally resets. For each user, at time \(t\) we maintain a current preference set \(F_t\subseteq\mathcal{I}\), where \(|\mathcal{I}|=N\) and \(|F_t|=m\). At each step, with probability \(\rho\) we reset \(F_t\) to \(m\) items sampled uniformly without replacement from \(\mathcal{I}\); we then generate the next item \(x_{t+1}\) by sampling uniformly from \(F_t\) with probability \(1-\varepsilon\), and sampling uniformly from \(\mathcal{I}\) with probability \(\varepsilon\).
Under the \(\mathrm{Oracle}\) setting, the predictor observes \(F_t\) at each time step, and \(\mathrm{Hit@1}^{\mathrm{Oracle}}\) can be derived analytically. In particular, when \(m=1\),
\begin{equation}
	\mathrm{Hit@1}^{\mathrm{Oracle}}=(1-\varepsilon)+\varepsilon/N.
	\label{eq:oracle_hit1_session_reset}
\end{equation}
\noindent We align difficulty by \(\mathrm{Hit@1}^{\mathrm{Oracle}}\) and invert Eq.~(\ref{eq:oracle_hit1_session_reset}) to obtain the noise parameter \(\varepsilon\). Here \(N\) characterizes the scale of the candidate space and is fixed to a large-scale setting, while \(m\) and \(\rho\) are held constant in our experiments.

\paragraph{Repeat Last generator.}
This process captures the simplest ``repeat the previous item'' mechanism. For each user, the initial item is sampled as \(x_1\sim \mathrm{Uniform}(\mathcal{I})\). For \(t\ge 1\), with probability \(p\) we set \(x_{t+1}=x_t\); otherwise, \(x_{t+1}\) is sampled uniformly from \(\mathcal{I}\).
Under the \(\mathrm{Oracle}\) setting, \(\mathrm{Hit@1}^{\mathrm{Oracle}}\) also admits a closed form:
\begin{equation}
	\mathrm{Hit@1}^{\mathrm{Oracle}}=p+(1-p)/N.
	\label{eq:oracle_hit1_repeat_last}
\end{equation}
\noindent We again align difficulty by \(\mathrm{Hit@1}^{\mathrm{Oracle}}\) and invert Eq.~(\ref{eq:oracle_hit1_repeat_last}) to obtain the noise parameter \(p\). The candidate-space scale \(N\) is fixed to a large-scale setting.

\paragraph{Results.}
We compare four predictability estimators by examining their deviations from the ideal reference ceiling \(\mathrm{Hit@1}^{\mathrm{Oracle}}\) under the two generators. Fig.~\ref{fig:synth_two_generators} reports the results, where the horizontal axis is the theoretical predictability and the vertical axis is the predictability estimate produced by each method; the gray dashed line indicates perfect agreement.
Under \emph{Session Reset}, \(\Pi_{\mathrm{Fano}}\) exhibits substantial variability in the low-predictability regime and deviates from the reference line overall (RMSE\(=0.395\)). Replacing the global candidate size with the reachability-based size yields \(\Pi_{\mathrm{Fano},N_r}\) with a trend closer to the reference and a lower error (RMSE\(=0.201\)), although systematic deviation and relatively large uncertainty remain. The permutation-entropy baseline produces a smoother curve, but the mapping to the target task is unstable (RMSE\(=0.236\)), especially in the low-predictability regime. In contrast, our method \(\Pi_{\mathrm{EPL}}=\exp(-\hat{S})\) attains the lowest RMSE (0.078) under this mechanism and varies monotonically with the difficulty control variable.
Under \emph{Repeat Last}, the two Fano-based estimates and the permutation-entropy baseline have limited ability to distinguish difficulty changes, producing near-saturated high predictability over a wide range (RMSEs of 0.278, 0.206, and 0.240, respectively). Our method aligns closely with the ideal reference curve under this mechanism (RMSE\(=0.027\)). Overall, the controlled synthetic experiments indicate that the proposed entropy-induced mapping provides a stable characterization of difficulty changes driven by noise parameters across different generative mechanisms, and yields lower errors than the baselines.

\subsection{Candidate Set Size Sensitivity}\label{sec:exp_n_sensitivity}
In sequential recommendation, the candidate-set size \(N\) is often shaped by evaluation protocols and retrieval-system configurations rather than being an intrinsic property of the task. To examine the sensitivity of different predictability estimators to \(N\), we adopt a controlled synthetic generator, \emph{Context Switch}, and sweep \(N\) across multiple orders of magnitude. The key idea is to invert the noise parameter so that the theoretical predictability remains constant under different \(N\), thereby attributing variations in estimator outputs primarily to their dependence on the candidate-set size.

\paragraph{Context Switch mechanism.}
This mechanism simulates a sequential process in which a user's interest context switches over time. Let \(\mathcal{I}\) be the global item set with \(|\mathcal{I}|=N\). We construct \(C\) contexts in advance, where each context corresponds to an item subset of size \(m_c\) sampled uniformly without replacement from \(\mathcal{I}\). For each user, at time \(t\) there is a current context \(c_t\). At each step, the context switches to a different one (distinct from the current context) with probability \(s\). The next item is then generated by sampling uniformly from the current context subset with probability \(1-\varepsilon\), and sampling uniformly from \(\mathcal{I}\) with probability \(\varepsilon\) as noise.
Conditional on the current context, the next-step distribution has two probability levels. Any item within the context subset has probability
\begin{equation}
	\frac{1-\varepsilon}{m_c}+\frac{\varepsilon}{N},
	\label{eq:context_switch_in_prob}
\end{equation}
while any item outside the subset has probability \(\varepsilon/N\). Therefore, under this setting, \(\mathrm{Hit@1}^{\mathrm{Oracle}}\) is directly determined by Eq.~(\ref{eq:context_switch_in_prob}) and does not depend on \(s\) or \(C\): these parameters affect the switching frequency and the resulting sequence structure, but do not change the form of the conditional distribution given the current context.

\paragraph{\(N\)-sweep design.}
We fix the theoretical predictability to \(\mathrm{Hit@1}^{\mathrm{Oracle}}=0.10\) and, for each \(N\) in the sweep, invert the corresponding \(\varepsilon(N)\) so that the same ideal reference ceiling is maintained across different candidate-space sizes. All other generation parameters are kept fixed: \(C=5\), \(m_c=5\), and \(s=0.05\). We generate sequences for \(300\) users with per-user length \(200\), and report the mean and error bars over multiple independent repetitions. For Fano-based baselines, the candidate size in the inversion equation is set to the current \(N\) being swept.

\begin{figure}[t]
	\centering
	\includegraphics[width=0.65\linewidth]{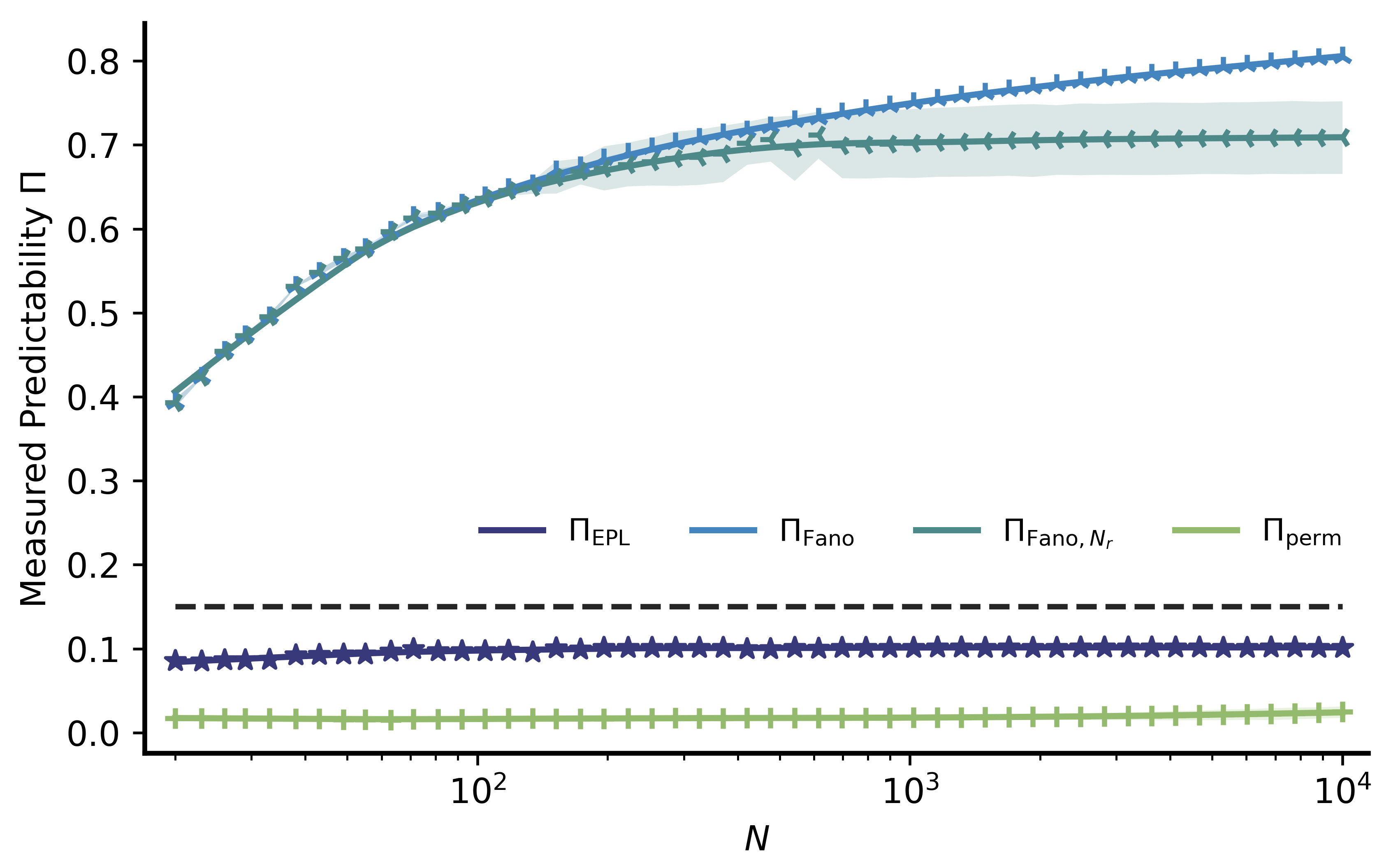}
	\caption{\(N\)-sweep results under the Context Switch generator. We fix \(\mathrm{Hit@1}^{\mathrm{Oracle}}=0.15\) and vary the candidate-set size \(N\) to compare how different predictability estimators change with \(N\).}
	\Description{N-sweep results on the Context Switch synthetic generator, comparing predictability estimators under a fixed ideal reference Hit@1 while varying candidate set size N.}
	\label{fig:context_switch}
\end{figure}

As shown in Fig.~\ref{fig:context_switch}, although the theoretical predictability is fixed for all \(N\), the Fano-based estimates still increase as \(N\) grows: \(\Pi_{\mathrm{Fano}}\) and \(\Pi_{\mathrm{Fano},N_r}\) rise from moderate values in the small-\(N\) regime and become much larger than the ideal reference ceiling in the large-\(N\) regime. This observation indicates that, even when the intrinsic difficulty of the underlying generative process is strictly controlled, the entropy--predictability inversion under the Fano route can be systematically re-calibrated by the candidate-space size, yielding estimates that drift with \(N\).
In contrast, our method \(\Pi_{\mathrm{EPL}}=\exp(-\hat{S})\) remains approximately constant with relatively small fluctuations throughout the sweep, reflecting robustness to changes in candidate-space size. In this mechanism, \(\Pi_{\mathrm{EPL}}\) is consistently lower than the ideal reference ceiling implied by the theoretical predictability, which corresponds to a more conservative difficulty characterization consistent with the lower-bound property in Theorem~\ref{thm:entropy_lower_bound}. The permutation-entropy baseline is also insensitive to \(N\), but its estimates remain close to zero and provide limited discrimination as a reference. Overall, the Context Switch \(N\)-sweep reveals the pronounced dependence of Fano-based estimators on the candidate-set size and further supports the stability of our method in large candidate spaces.

\subsection{Real-world Validation: Dataset-level Ranking Consistency}\label{sec:exp_real}

\begin{table}[t]
	\caption{Best model performance on real-world datasets: for each dataset, we select the model with the highest test hit rate from the candidate algorithm set and report \(\mathrm{Hit@1}\) and \(\mathrm{Hit@20}\).}
	\label{tab:best_hit20_real}
	\centering
	\begin{tabular*}{0.8\linewidth}{@{\extracolsep{\fill}} l l r r}
		\toprule
		Dataset & Best model & Test \(\mathrm{Hit@1}\) & Test \(\mathrm{Hit@20}\)\\
		\midrule
		AOTM & SASRec & 0.0000 & 0.1485\\
		Delicious & BERT4Rec & 0.0000 & 0.0217\\
		Online Retail & SASRec & 0.0271 & 0.3254\\
		Personality & SHAN & 0.0000 & 0.1818\\
		LastFM & FOSSIL & 0.0926 & 0.2534\\
		TaFeng & FOSSIL & 0.0090 & 0.1317\\
		MovieLens-100K & LightSANs & 0.0180 & 0.2460\\
		MovieLens-1M & GRU4Rec & 0.0677 & 0.3874\\
		MovieLens-20M & GRU4Rec & 0.0482 & 0.2867\\
		Algebra & BERT4Rec & 0.4146 & 0.7317\\
		Bridge & GRU4Rec & 0.4546 & 0.6798\\
		\bottomrule
	\end{tabular*}
\end{table}

\begin{figure}[t]
	\centering
	\begin{minipage}{0.49\linewidth}
		\centering
		\includegraphics[width=\linewidth]{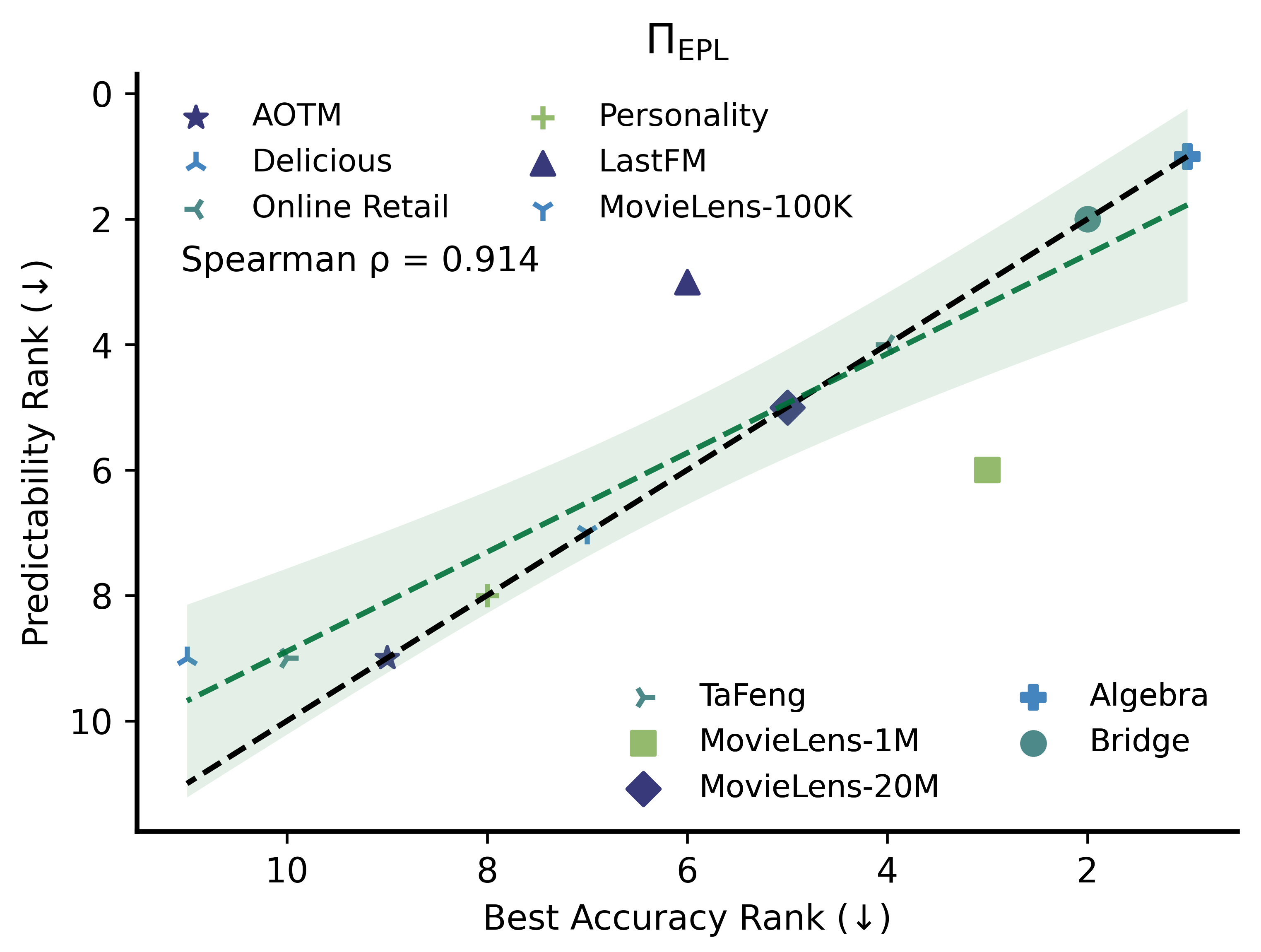}
	\end{minipage}\hfill
	\begin{minipage}{0.49\linewidth}
		\centering
		\includegraphics[width=\linewidth]{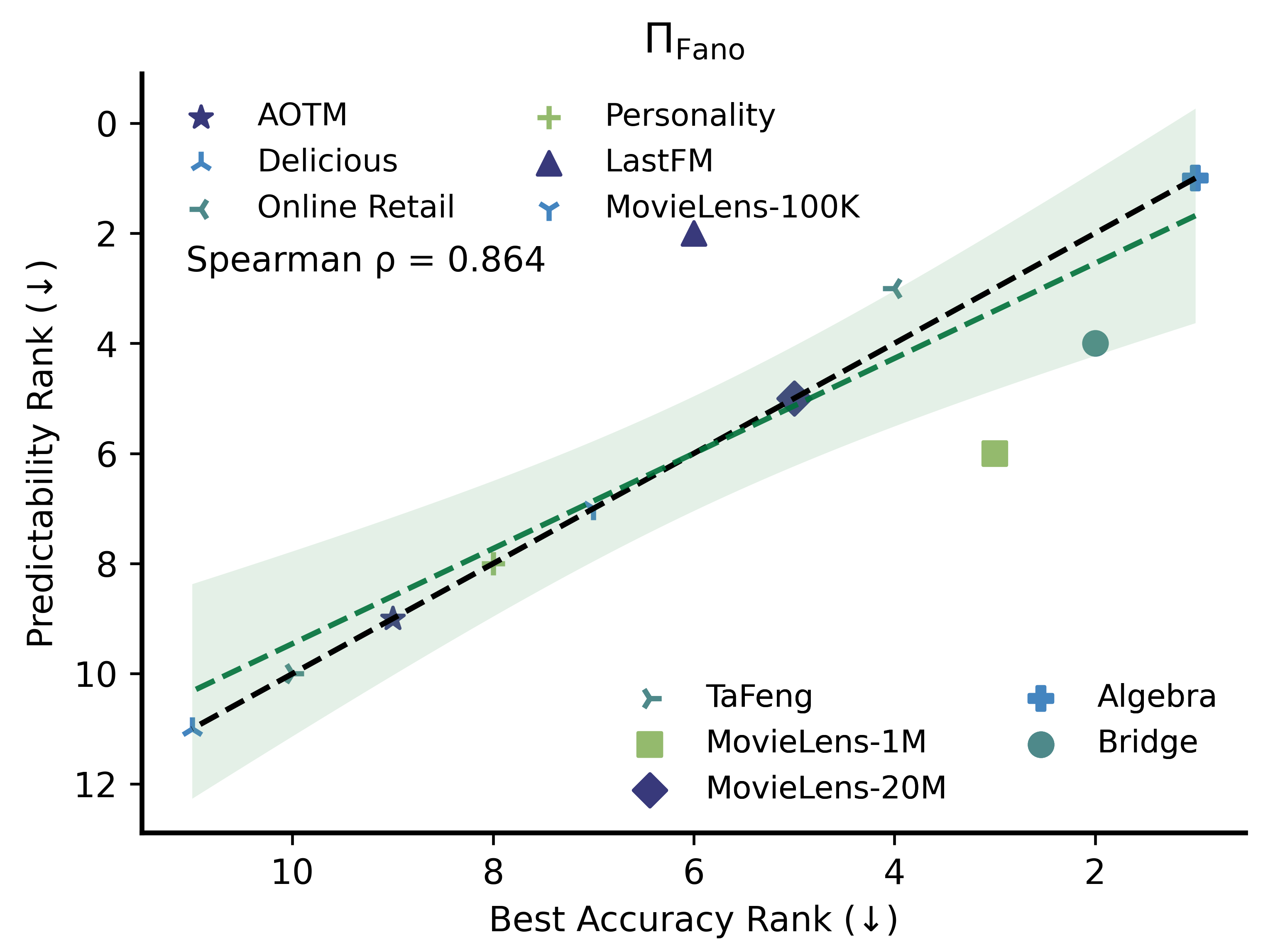}
	\end{minipage}
	\caption{Ranking consistency on real-world datasets. The horizontal axis shows the rank of the best model's \(\mathrm{Hit@20}\), and the vertical axis shows the rank of a predictability metric; the diagonal indicates perfect agreement. Left: \(\Pi_{\mathrm{EPL}}\). Right: \(\Pi_{\mathrm{Fano}}\).}
	\Description{Rank consistency between dataset-level predictability and best model Hit@20 on real datasets, shown for the proposed EPL estimator and the Fano baseline.}
	\label{fig:rank_consistency_real}
\end{figure}

\begin{figure}[t]
	\centering
	\begin{minipage}{0.49\linewidth}
		\centering
		\includegraphics[width=\linewidth]{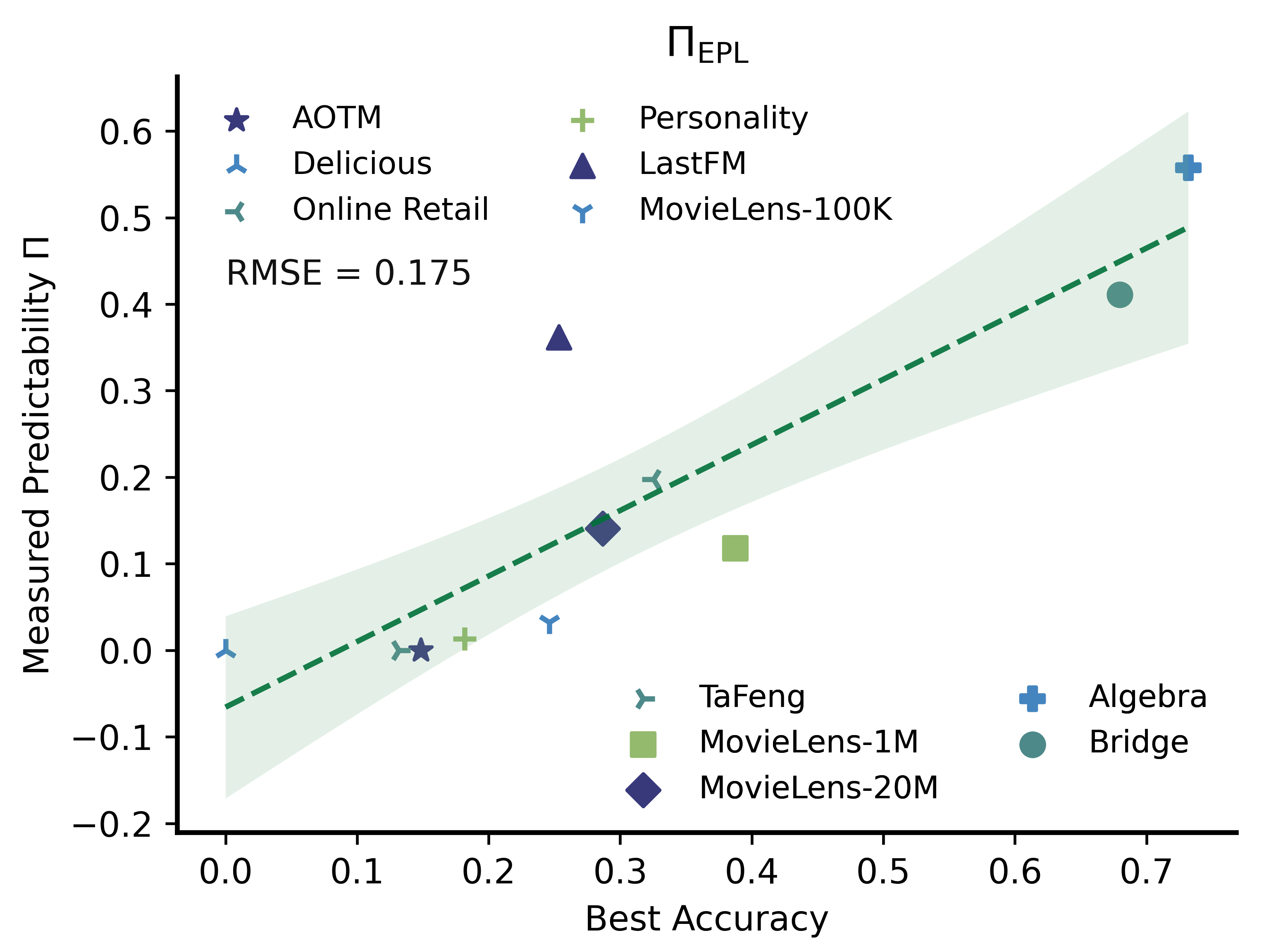}
	\end{minipage}\hfill
	\begin{minipage}{0.49\linewidth}
		\centering
		\includegraphics[width=\linewidth]{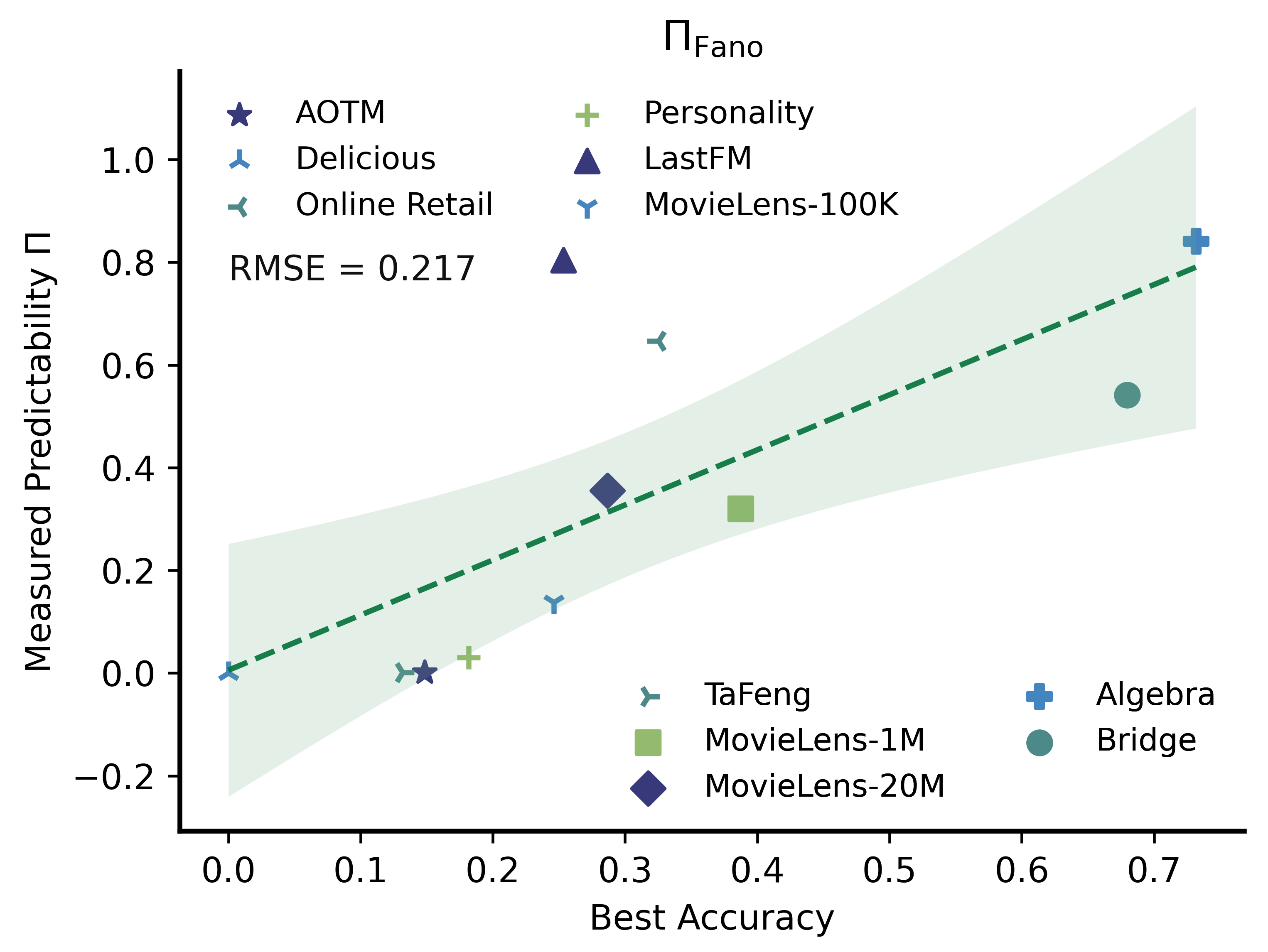}
	\end{minipage}
	\caption{Value-level comparison on real-world datasets. The horizontal axis shows the best model's test \(\mathrm{Hit@20}\), and the vertical axis shows the corresponding predictability estimate. Left: \(\Pi_{\mathrm{EPL}}\). Right: \(\Pi_{\mathrm{Fano}}\).}
	\Description{Value-level comparison between dataset-level predictability estimates and best model Hit@20 on real datasets, shown for the proposed EPL estimator and the Fano baseline.}
	\label{fig:value_consistency_real}
\end{figure}
On real-world datasets, we train and tune a set of representative sequential recommendation algorithms under a unified training and evaluation protocol, and use the best test hit rate as a reference for the attainable offline performance under this protocol. In large candidate spaces, \(\mathrm{Hit@1}\) is often extremely low or even zero and thus provides limited discriminative information. Following the classical setting for predictability limits in session-based next-item recommendation~\cite{jarv2019predictability_session_nextitem}, we use \(\mathrm{Hit@20}\) as the primary accuracy metric for comparison. Table~\ref{tab:best_hit20_real} summarizes the best models' \(\mathrm{Hit@1}\) and \(\mathrm{Hit@20}\) on each dataset.
To examine whether predictability estimates provide a consistent characterization of the relative difficulty across real-world datasets, we analyze both rank consistency and value-level agreement. Specifically, for each dataset \(d\), we define the best-model accuracy \(A_d\) as the best \(\mathrm{Hit@20}\) in Table~\ref{tab:best_hit20_real}. We also define a dataset-level predictability score \(P_d\) by estimating entropy on the sequences and aggregating the resulting predictability values into a single dataset-level number via a weighted aggregation. We first compare the ranking consistency between \(\{A_d\}\) and \(\{P_d\}\) using Spearman's rank correlation, as shown in Fig.~\ref{fig:rank_consistency_real}. We then compare \(\{A_d\}\) and \(\{P_d\}\) directly at the value level to reveal potential headroom across datasets, as shown in Fig.~\ref{fig:value_consistency_real}.

As shown in Fig.~\ref{fig:rank_consistency_real}(a), \(\Pi_{\mathrm{EPL}}\) exhibits a high agreement between the dataset ranking induced by predictability and the ranking induced by the best model, with a Spearman rank correlation of \(0.914\); most datasets lie close to the diagonal. In comparison, Fig.~\ref{fig:rank_consistency_real}(b) shows a lower consistency for \(\Pi_{\mathrm{Fano}}\), with Spearman correlation \(0.864\). The remaining baselines are weaker, with Spearman correlations of \(0.536\) for \(\Pi_{\mathrm{Fano},N_r}\) and \(-0.10\) for \(\Pi_{\mathrm{perm}}\). These results indicate that \(\Pi_{\mathrm{EPL}}\) can produce a more consistent ranking of relative dataset difficulty without training recommendation models.
At the value level, Fig.~\ref{fig:value_consistency_real} further shows the correspondence between predictability estimates and the best \(\mathrm{Hit@20}\). Compared with \(\Pi_{\mathrm{Fano}}\), our method exhibits a more stable relationship with the best \(\mathrm{Hit@20}\) and a smaller overall error: the RMSE of \(\Pi_{\mathrm{EPL}}\) is \(0.175\), lower than the RMSE of \(\Pi_{\mathrm{Fano}}\) (\(0.217\)). This result suggests that, on real-world datasets, \(\Pi_{\mathrm{EPL}}\) provides not only a more consistent difficulty ranking but also a more reliable numerical reference for attainable performance.
In addition, both Fig.~\ref{fig:rank_consistency_real} and Fig.~\ref{fig:value_consistency_real} show that LastFM exhibits a larger gap between its predictability level and its best \(\mathrm{Hit@20}\), suggesting substantial headroom and motivating further algorithm design tailored to its sequential structure.

\subsection{Predictability-Guided Training Data Selection}\label{sec:exp_data_selection}
The preceding experiments mainly position predictability as a model-agnostic characterization of intrinsic task attainability, enabling analyses of relative difficulty and headroom across datasets and user groups. A more practical question, however, is whether predictability can also be used proactively to guide data construction. In particular, under constrained training budgets, if user-level predictability estimates are already available, can they help identify additional training data that are more valuable for downstream model learning?

To answer this question, we design a controlled data-selection protocol. We first partition long-sequence users into two disjoint sets, denoted \texttt{eval users} and \texttt{candidate users}. The \texttt{eval users} define a fixed evaluation task: for each \texttt{eval user}, the historical prefix is included in the base training set, while the last next-item instance is held out for testing. As a result, all compared strategies share the same test set, the same evaluation item space, and the same base training task. The \texttt{candidate users} serve only as a pool of additional training resources. We compute user-level predictability scores only on \texttt{candidate users} and construct three extra-data selection strategies accordingly: selecting users with high predictability (\texttt{HighPi}), selecting users uniformly at random (\texttt{Random}), and selecting users with low predictability (\texttt{LowPi}). Under each extra-data budget, the amount of additional training data is kept the same across strategies, so the comparison is strictly restricted to the following question: under the same evaluation task, the same base training set, and the same extra-data budget, do different types of additional training data contribute differently to model performance?

This protocol is important because it separates two distinct questions: whether a user group is itself easier to predict, and whether that group provides more useful training signals for another fixed target task. Since \texttt{eval users}, test instances, and the evaluation space are all held fixed across strategies, the comparison is not about differences in test difficulty across user groups. Rather, it measures how additional training data drawn from different user groups affect performance on the same downstream task. Therefore, this experiment evaluates predictability as a signal of training-data value, rather than merely re-confirming that high-predictability users are easier to model.

\begin{figure*}[t]
	\centering
	\includegraphics[width=0.75\textwidth]{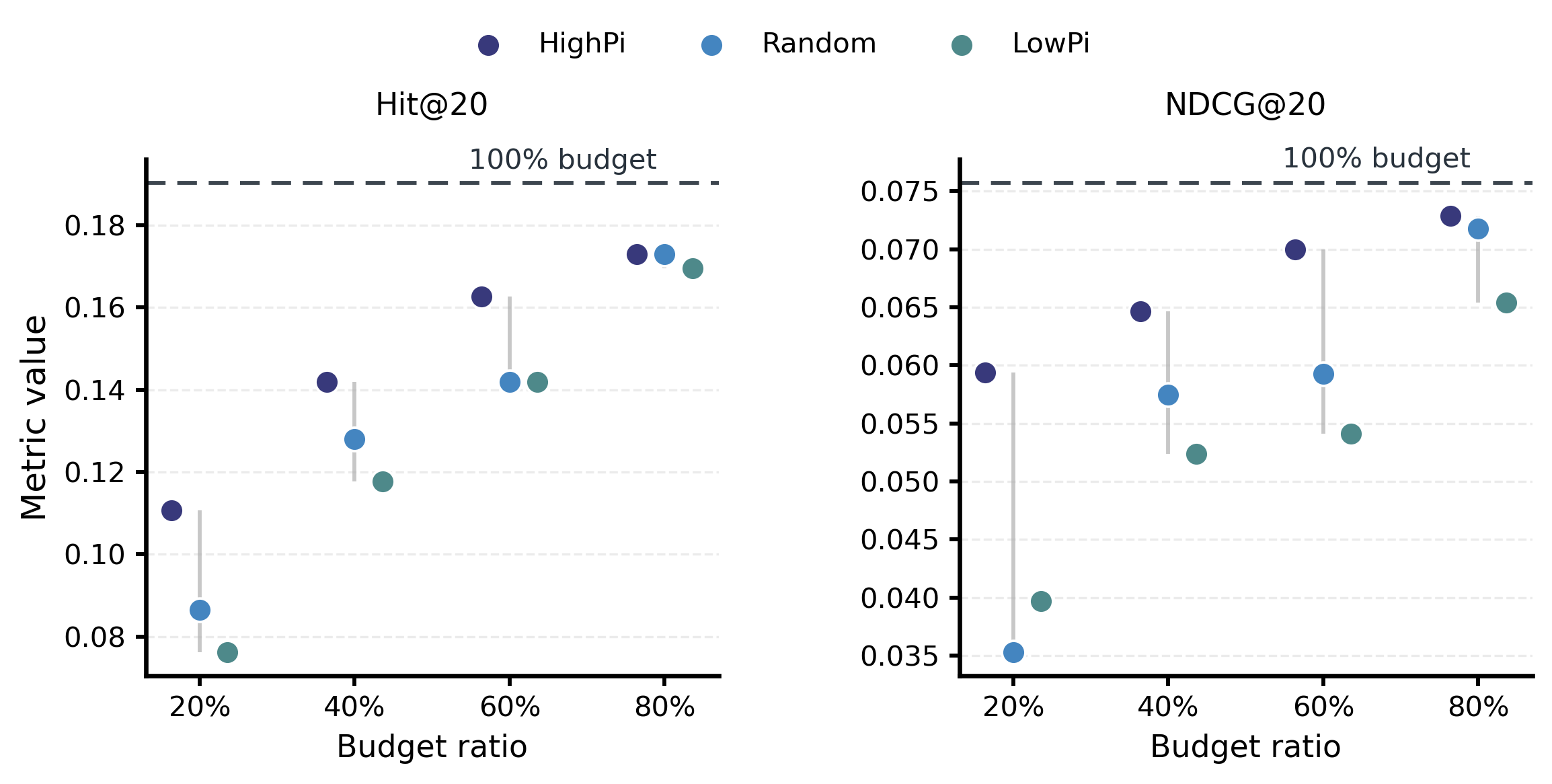}
	\caption{Predictability-guided training data selection under different extra-data budgets. Under a fixed evaluation task, a fixed base training set, and equal additional-data budgets, we compare three strategies for selecting extra training users from the candidate pool: \texttt{HighPi}, \texttt{Random}, and \texttt{LowPi}. The left and right panels report \(\mathrm{Hit@20}\) and \(\mathrm{NDCG@20}\), respectively. The dashed line denotes the corresponding result when all candidate-user data are included.}
	\Description{Predictability-guided training data selection under different extra-data budgets, comparing HighPi, Random, and LowPi on Hit@20 and NDCG@20, with a dashed line indicating the full-data reference.}
	\label{fig:pi_guided_data_selection}
\end{figure*}

Figure~\ref{fig:pi_guided_data_selection} reports the results under different budget ratios. Across both \(\mathrm{Hit@20}\) and \(\mathrm{NDCG@20}\), the \texttt{HighPi} strategy consistently outperforms \texttt{Random} and \texttt{LowPi}, yielding an overall ordering of \texttt{HighPi > Random > LowPi}; this advantage is particularly pronounced in the low-budget regime. These results indicate that, under a fixed evaluation task and a fixed additional-data budget, extra training samples contributed by high-predictability users are more valuable for downstream learning. Taken together, the findings suggest that predictability is not only a model-agnostic characterization of task attainability, but also a practically useful signal for training-data construction and budget-aware resource allocation.

\subsection{Analysis by Novelty Preference}\label{sec:exp_novelty}
Beyond dataset-level difficulty, sequential recommendation often exhibits substantial heterogeneity at the user level. To quickly identify which users are intrinsically easier to predict without training models, we perform group-based predictability diagnosis. We use users' novelty preference as the grouping dimension. For an item \(i\), we define its popularity \(\text{pop}(i)\) as its relative frequency in all interactions, and measure its novelty by \(-\log \text{pop}(i)\). A user's novelty preference is then defined as the average novelty of items in the user's historical interactions; a larger value indicates a stronger preference for less popular and more novel items. Within each dataset, we split users into two groups based on this score, denoted Q1 and Q2, corresponding to lower and higher novelty preference, and compute predictability estimates within each group.

Fig.~\ref{fig:cohort_novelty_pi} shows predictability differences across novelty-preference groups and datasets. Overall, the low-novelty-preference group (Q1) tends to have higher predictability than the high-novelty-preference group (Q2), suggesting that interaction sequences with more popular and more repetitive choices exhibit stronger regularities and are therefore easier to predict in a model-agnostic sense. In contrast, users with higher novelty preference interact with a more diverse set of items and exhibit stronger exploration, leading to higher uncertainty and lower predictability.
To examine whether this group difference is reflected in observed recommendation performance, we further report the best model's offline \(\mathrm{Hit@20}\) under the same user-group partition, as shown in Fig.~\ref{fig:cohort_novelty_hit20}. The gap in \(\mathrm{Hit@20}\) between the two groups is consistent with the predictability difference, indicating that predictability estimates can effectively characterize relative attainability across user groups. Together with the dataset-level ranking consistency results in Section~\ref{sec:exp_real}, these findings suggest that our method enables rapid within-dataset difficulty profiling, supporting practical diagnostics for dataset selection, user-group analysis, and the choice of model families.

\begin{figure*}[t]
	\centering
	\includegraphics[width=\textwidth]{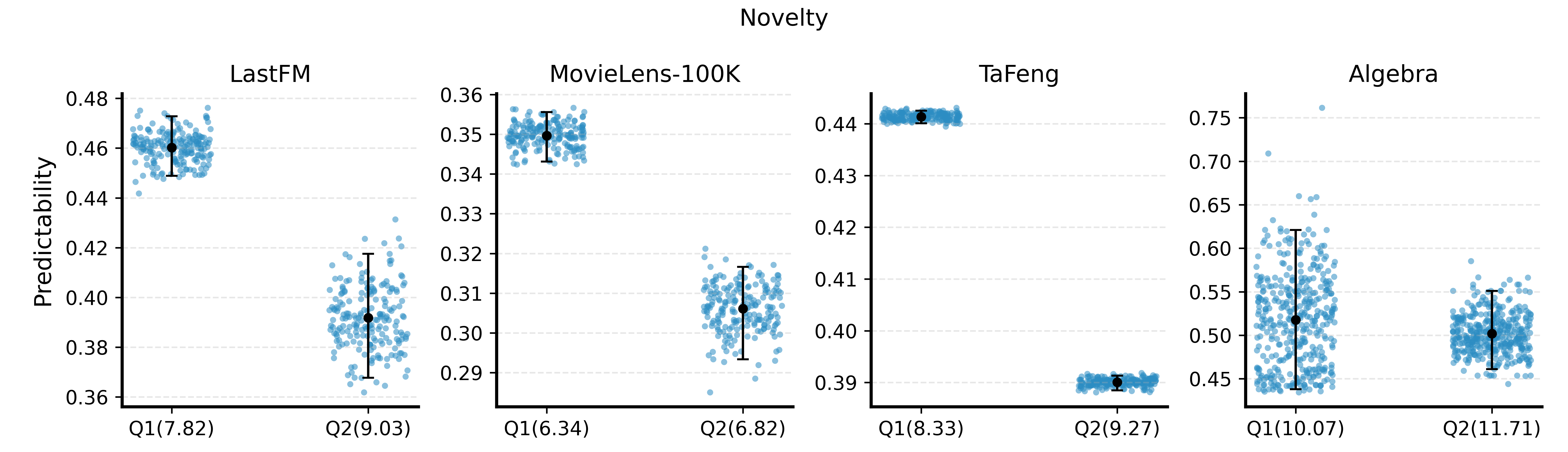}
	\vspace{0.5em}
	\includegraphics[width=\textwidth]{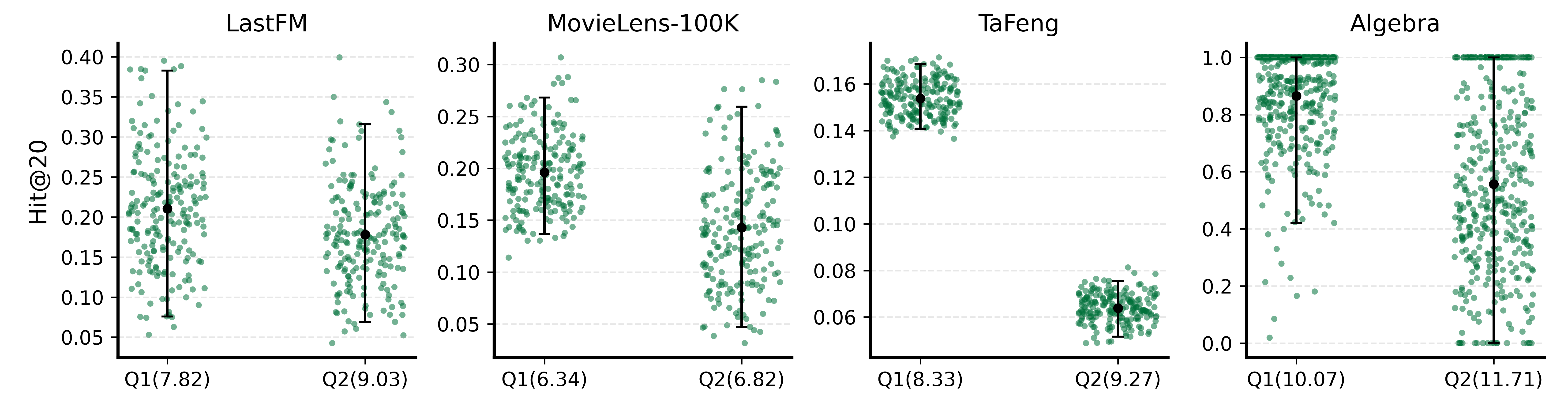}
	\caption{Analysis by novelty preference. The top panel reports predictability differences between user groups with lower (Q1) and higher (Q2) novelty preference, and the bottom panel reports the corresponding differences in the best model's \(\mathrm{Hit@20}\), assessing whether group-level predictability aligns with observable recommendation performance. Dots denote user-level estimates; black points and error bars show within-group aggregates and uncertainty.}
	\Description{Novelty-preference group results across datasets: predictability (top) and best-model Hit@20 (bottom).}
	\label{fig:cohort_novelty_pi}
	\label{fig:cohort_novelty_hit20}
\end{figure*}

\subsection{Analysis by Long-tail Exposure}\label{sec:exp_longtail}
In addition to novelty preference, users' exposure to long-tail items can also affect the intrinsic difficulty of sequential recommendation. We quantify long-tail exposure by the fraction of a user's interactions that fall into a globally low-frequency item set, and partition users into two groups, Q1 and Q2, corresponding to lower and higher long-tail exposure. Fig.~\ref{fig:cohort_longtail_pi} reports the resulting group-level predictability across datasets. The low-exposure group tends to exhibit higher predictability, whereas the high-exposure group shows lower predictability with larger variability. This pattern suggests that long-tail consumption is associated with both lower attainable accuracy and greater heterogeneity in behavioral regularities, making further performance improvements more dependent on effectively leveraging sparse signals and long-tail structures. Overall, this grouped analysis further highlights the value of predictability estimation as a fast diagnostic tool, enabling the identification of user regions whose difficulty is driven by long-tail exposure and informing algorithm selection and system design.

\begin{figure*}[t]
	\centering
	\includegraphics[width=\textwidth]{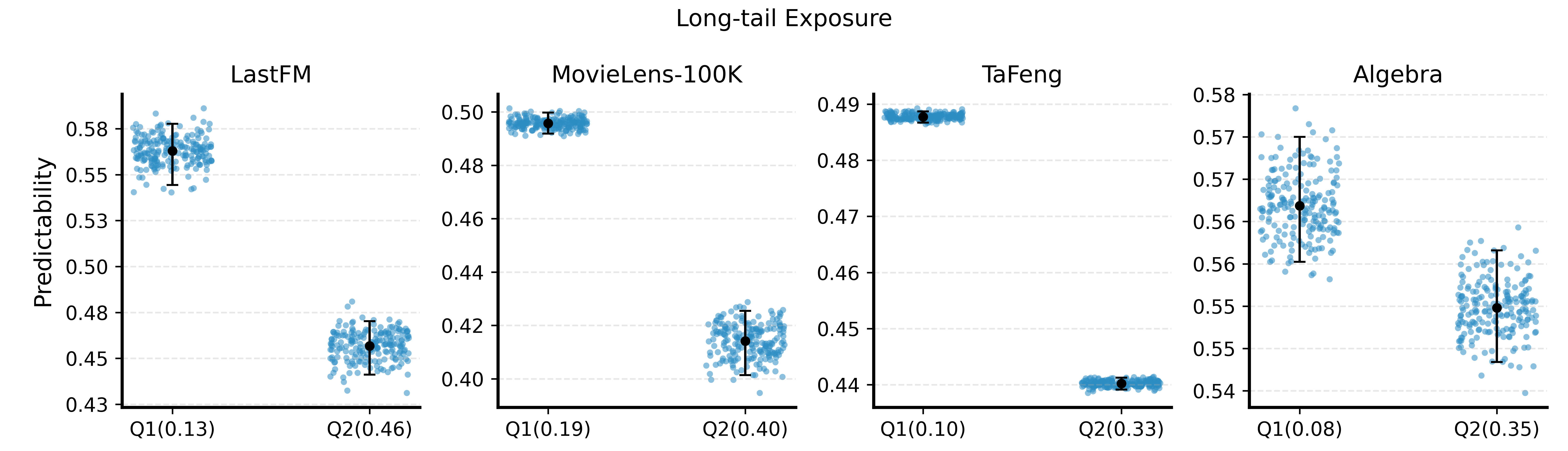}
	\caption{Analysis by long-tail exposure. Users are grouped by the fraction of long-tail items in their interaction histories, and predictability is compared across groups. Q1/Q2 correspond to lower/higher long-tail exposure; dots denote user-level estimates, and black points with error bars indicate within-group aggregates and uncertainty.}
	\Description{Predictability under long-tail exposure groups across several datasets.}
	\label{fig:cohort_longtail_pi}
	\vspace{-1.2em}
\end{figure*}

\subsection{Analysis by User Activity}\label{sec:exp_activity}
User activity directly determines the observable sequence length and the density of behavioral signals, and thus can affect predictability. We measure activity by the number of interactions per user and partition users into two groups, Q1 and Q2, corresponding to higher and lower activity. Fig.~\ref{fig:cohort_activity_pi} shows the group-wise results across datasets. Overall, the high-activity group exhibits higher predictability. This can be attributed to two factors: longer histories make stable preferences more observable and reduce estimation variance, and highly active users tend to have more concentrated interest distributions and stronger repetitive regularities. These results indicate that predictability estimation can quickly identify user regions that are both low-activity and hard to predict at a very low computational cost, providing actionable guidance for strategy design and model selection in cold-start and sparse regimes.

\begin{figure*}[t]
	\centering
	\includegraphics[width=\textwidth]{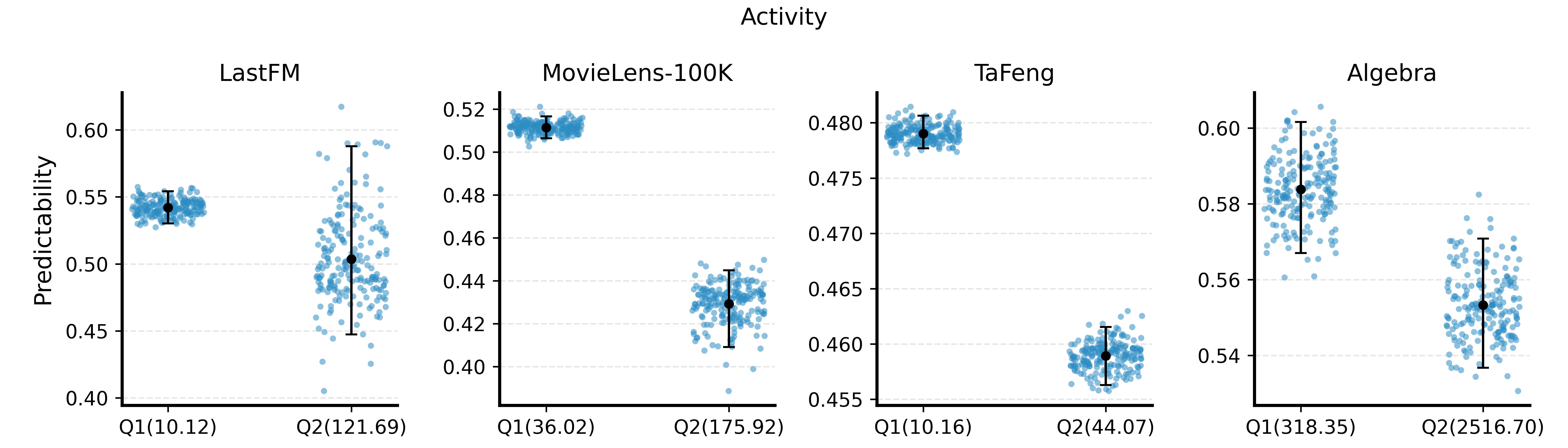}
	\caption{Analysis by user activity. Users are grouped by their number of interactions to examine how behavioral-signal density affects predictability. Q1/Q2 denote the higher/lower activity groups; dots indicate user-level estimates, and black points with error bars show within-group aggregates and uncertainty.}
	\Description{Predictability under user activity groups across several datasets.}
	\label{fig:cohort_activity_pi}
\end{figure*}